\pdfoutput=1
\RequirePackage{ifpdf}
\ifpdf 
\documentclass[pdftex]{sigma}
\else
\documentclass{sigma}
\fi

\numberwithin{equation}{section}

\newtheorem{Theorem}{Theorem}[section]
\newtheorem*{Theorem*}{Theorem}
\newtheorem{Corollary}[Theorem]{Corollary}
\newtheorem{Lemma}[Theorem]{Lemma}
\newtheorem{Proposition}[Theorem]{Proposition}
 { \theoremstyle{definition}
\newtheorem{Definition}[Theorem]{Definition}

\newtheorem{Example}[Theorem]{Example}
\newtheorem{Remark}[Theorem]{Remark} }

\newcommand{\p}{\partial}
\newcommand{\de}{\mathrm{d}}
\newcommand{\bbr}{\mathbb{R}}
\newcommand{\bbz}{\mathbb{Z}}
\newcommand{\s}{\mathrm{S}}
\newcommand{\id}{\mathrm{id}}
\newcommand{\D}{\mathrm{D}}
\newcommand{\E}{\mathrm{E}}
\newcommand{\Lrm}{\mathrm{L}}
\newcommand{\uJalpha}{u_{\bf J}^{\alpha}}

\newcommand{\mbn}{\mathbf{n}}
\newcommand{\mbu}{\mathbf{u}}
\newcommand{\mbv}{\mathbf{v}}
\newcommand{\mbx}{\mathbf{x}}
\newcommand{\mbI}{\mathbf{I}}
\newcommand{\mbJ}{\mathbf{J}}
\newcommand{\mbK}{\mathbf{K}}
\newcommand{\mbzero}{\mathbf{0}}
\newcommand{\mcH}{\mathcal{H}}
\newcommand{\mcD}{\mathcal{D}}
\newcommand{\mcJ}{\mathcal{J}}
\newcommand{\mcK}{\mathcal{K}}

\newcommand{\mcT}{\mathcal{T}}

\newcommand{\pde}{P$\Delta$E\ }
\newcommand{\pdes}{P$\Delta$Es\ }
\newcommand{\pdee}{P$\Delta$E}
\newcommand{\pdese}{P$\Delta$Es}
\newcommand{\dde}{D$\Delta$E\ }
\newcommand{\ddes}{D$\Delta$Es\ }

\newcommand{\ddese}{D$\Delta$Es}
\newcommand{\prol}{P(J^\infty(\mcT_\mbn))}
\newcommand{\uoo}{u_{0;0}}
\newcommand{\uoi}{u_{0;1}}
\newcommand{\uio}{u_{1;0}}

\newcommand{\uom}{u_{0;-1}}

\newcommand{\voo}{v_{0;0}}
\newcommand{\voi}{v_{0;1}}
\newcommand{\vio}{v_{1;0}}

\newcommand{\vom}{v_{0;-1}}

\usepackage{tikz}
\usetikzlibrary{arrows,arrows.meta,decorations.markings}

\usepackage{subcaption}
\usepackage[labelsep=period,labelfont=bf,font=small]{caption}

\begin{document}
\allowdisplaybreaks

\renewcommand{\thefootnote}{}

\newcommand{\arXivNumber}{2309.09040}

\renewcommand{\PaperNumber}{006}

\FirstPageHeading

\ShortArticleName{Moving Frames: Difference and Differential-Difference Lagrangians}

\ArticleName{Moving Frames: Difference and Differential-Difference\\ Lagrangians\footnote{This paper is a~contribution to the Special Issue on Symmetry, Invariants, and their Applications in honor of Peter J.~Olver. The~full collection is available at \href{https://www.emis.de/journals/SIGMA/Olver.html}{https://www.emis.de/journals/SIGMA/Olver.html}}}

\Author{Lewis C.~WHITE and Peter E.~HYDON}
\AuthorNameForHeading{L.C.~White and P.E.~Hydon}
\Address{School of Mathematics, Statistics and Actuarial Science, University of Kent,\\
Canterbury, Kent, CT2 7NF, UK}
\Email{\href{mailto:lcwhite29@gmail.com}{lcwhite29@gmail.com}, \href{mailto:P.E.Hydon@kent.ac.uk}{P.E.Hydon@kent.ac.uk}}

\ArticleDates{Received September 19, 2023, in final form January 09, 2024; Published online January 15, 2024}

\Abstract{This paper develops moving frame theory for partial difference equations and for differential-difference equations with one continuous independent variable. In each case, the theory is applied to the invariant calculus of variations and the equivariant formulation of the conservation laws arising from Noether's theorem. The differential-difference theory is not merely an amalgam of the differential and difference theories, but has additional features that reflect the need for the group action to preserve the prolongation structure. Projectable moving frames are introduced; these cause the invariant derivative operator to commute with shifts in the discrete variables. Examples include a Toda-type equation and a method of lines semi-discretization of the nonlinear Schr\"odinger equation.}


\Classification{39A14; 58D19; 47E07}

\renewcommand{\thefootnote}{\arabic{footnote}}
\setcounter{footnote}{0}

\section{Introduction}

The modern formulation of moving frames introduced by Fels and Olver \cite{fels1998moving,fels1999moving} is a powerful tool in the analysis of partial differential equations (PDEs). It enables one to reduce a given system of PDEs to an invariant system by factoring out Lie symmetry group orbits (locally, at least).
If the symmetries are extraneous to the problem of interest (for instance, projective symmetries in computer vision \cite{olvercompvis}), the moving frame provides a major simplification. For a~clear, straightforward introduction to moving frames, see Mansfield's text \cite{mansfield2010practical}.

Many PDE systems of interest are Euler--Lagrange equations with a Lie group of variational symmetries. For such systems, moving frames are most effectively applied by invariantizing the Lagrangian functional directly \cite{kogan2003invariant}. With this approach, Noether's theorem has an elegant formulation in terms of the adjoint action of the Lie group on a set of invariants \cite{gonccalves2012moving,gonccalves2013moving,gonccalves2016moving}.

Difference equations have discrete independent variables, so any Lie symmetries act only on the dependent variables. Commonly, the action varies with the discrete variables. Moving frames can be adapted to an equation on a finite set of points by using a finite-dimensional product space. Finite difference approximation of a given differential equation requires consistency as the points coalesce. This constraint led to the introduction of multi-space \cite{olvermulti}, which has been used to construct (highly accurate) invariant approximations of ordinary differential equations (ODEs) \cite{kimolver}. More generally, a discrete moving frame attaches a finite-dimensional product space to each base point, without imposing coalescence or any other structure in advance \cite{mansfield2013discrete}. The discrete moving frame construction works for any number of independent variables and has been used to generalize multi-space to higher dimensions \cite{maribeffa2018discrete}.

Difference equations have an intrinsic structure that arises from their mesh point labels. We restrict attention to the most common case, an $m$-dimensional logically rectangular mesh. The labels $n^i\in\bbz,\ i=1,\dots, m$, can be regarded as the independent variables. Each label belongs to an ordered set, so it is helpful to incorporate this ordering into the moving frame definition. This has been achieved for ordinary difference equations (O$\Delta$Es), resulting in the invariant variational calculus and equivariant Noether's theorem \cite{mansfield2019movinga,mansfield2019movingb}. The current paper extends this approach to partial difference equations (P$\Delta$Es) and to differential-difference equations (D$\Delta$Es) with one continuous independent variable.

Section \ref{2sec} summarizes the building-blocks of the \pde theory, from which we develop difference moving frames (see Section \ref{3sec}), the invariant calculus of variations (see Section \ref{Section: The invariant formulation of the Euler--Lagrange equations}), and the equivariant formulation of Noether's conservation laws (see Section \ref{Section: Conservation laws}). Even for scalar \pdese, it is necessary to use several generating invariants and to take the relations between these into account. The more dependent and independent variables there are, the more complex these relations can become. For clarity, we illustrate the general theory with fairly straightforward examples. In particular, we use a Toda-type equation as a running example to show the various aspects of the theory.

There is one aspect of the \pde theory that is simpler than its counterpart for PDEs: the Lie group action on the independent variables is trivial. This is not necessarily the case for \ddese, although there are constraints on the group action, as discussed in Section \ref{Section: Differential-difference structure}. These constraints suggest the idea of a projectable moving frame, which is a major simplification. Section \ref{Section: The differential-difference calculus of variations} outlines the invariant variational calculus and Noether conservation laws for \ddese, emphasizing those aspects of the theory that do not follow immediately from the \pde theory. Section \ref{exam} presents some examples, including a method of lines semi-discretization of the nonlinear Schr\"odinger equation. Concluding remarks are given in Section \ref{conc}.

\section{A brief summary of the building-blocks}\label{2sec}

\subsection{Difference prolongation space} \label{Section: Difference prolongation space}

A given differential equation can be represented as a variety within an appropriate jet space (see~\cite{olver2000applications}). A difference equation has discrete independent variables, so to use moving frames, one must represent the equation as a variety within an appropriate continuous space. Such spaces are subspaces of the difference prolongation space \cite{peng2022transformations}, which is described briefly in this section.

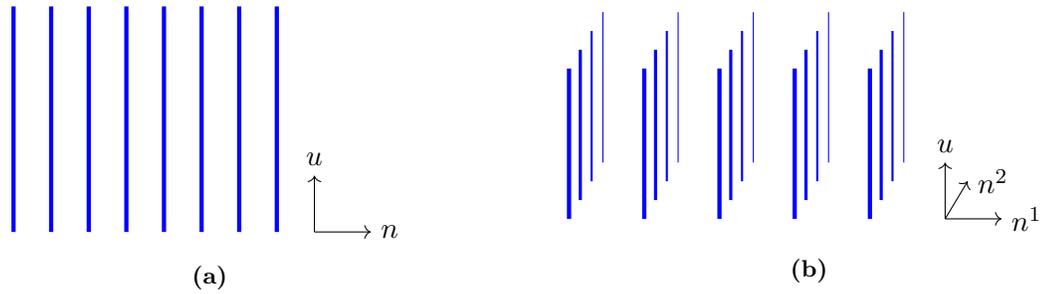
\begin{figure}[t]
\centering
 \begin{subfigure}[h]{0.49\textwidth}
 \centering
 \begin{tikzpicture}
 \draw[->] (4,0) -- (4,0.75) node[above] {$u$};
 \draw[->] (4,0) -- (4.75,0) node[right] {$n$};
 \draw[blue, ultra thick] (0,0) -- (0,3);
 \draw[blue, ultra thick] (0.5,0) -- (0.5,3);
 \draw[blue, ultra thick] (1,0) -- (1,3);
 \draw[blue, ultra thick] (1.5,0) -- (1.5,3);
 \draw[blue, ultra thick] (2,0) -- (2,3);
 \draw[blue, ultra thick] (2.5,0) -- (2.5,3);
 \draw[blue, ultra thick] (3,0) -- (3,3);
 \draw[blue, ultra thick] (3.5,0) -- (3.5,3);
 \end{tikzpicture}
 \caption{}
 \label{Fig: Ordinary difference equation}
 \end{subfigure}
\begin{subfigure}[h]{0.49\textwidth}
\centering
 \begin{tikzpicture}
 \draw[->] (5,0) -- (5,0.75) node[above] {$u$};
 \draw[->] (5,0) -- (5.75,0) node[right] {$n^{1}$};
 \draw[->] (5,0) -- (5.3,0.5) node[right] {$n^{2}$};
 \draw[blue, ultra thick] (0,0) -- (0,2);
 \draw[blue, very thick] (0.15,0.25) -- (0.15,2.25);
 \draw[blue, thick] (0.3,0.5) -- (0.3,2.5);
 \draw[blue, thin] (0.45,0.75) -- (0.45,2.75);
 \draw[blue, ultra thick] (1,0) -- (1,2);
 \draw[blue, very thick] (1.15,0.25) -- (1.15,2.25);
 \draw[blue, thick] (1.3,0.5) -- (1.3,2.5);
 \draw[blue, thin] (1.45,0.75) -- (1.45,2.75);
 \draw[blue, ultra thick] (2,0) -- (2,2);
 \draw[blue, very thick] (2.15,0.25) -- (2.15,2.25);
 \draw[blue, thick] (2.3,0.5) -- (2.3,2.5);
 \draw[blue, thin] (2.45,0.75) -- (2.45,2.75);
 \draw[blue, ultra thick] (3,0) -- (3,2);
 \draw[blue, very thick] (3.15,0.25) -- (3.15,2.25);
 \draw[blue, thick] (3.3,0.5) -- (3.3,2.5);
 \draw[blue, thin] (3.45,0.75) -- (3.45,2.75);
 \draw[blue, ultra thick] (4,0) -- (4,2);
 \draw[blue, very thick] (4.15,0.25) -- (4.15,2.25);
 \draw[blue, thick] (4.3,0.5) -- (4.3,2.5);
 \draw[blue, thin] (4.45,0.75) -- (4.45,2.75);
\end{tikzpicture}
\caption{}
\label{Fig: Partial difference equation}
\end{subfigure}
\caption{(a) The total space for scalar O$\Delta$Es is $\mathcal{T}=\bbz \times \bbr$; (b) the total space for scalar P$\Delta$Es with two independent variables is $\mathcal{T} = \bbz^{2} \times \bbr$.} \label{Fig: Difference equations graph}
\end{figure}

Consider a difference equation with independent variables ${\bf{n}} := \bigl(n^{1},n^{2},\dots,n^{m}\bigr) \in \bbz^{m}$ and dependent variables ${\bf{u}} := \bigl(u^{1},u^{2},\dots,u^{q}\bigr) \in \bbr^{q}$.
These variables are coordinates on the \textit{total space}~$\bbz^{m} \times \bbr^{q} $; a solution of the \pde is a graph on the total space. Figure \ref{Fig: Difference equations graph} illustrates the total spaces for a scalar O$\Delta$E and \pdee. For simplicity, we assume that the equation holds for all $\mbn\in\bbz^m$; our results apply \textit{mutatis mutandis} to difference equations on any product lattice (see \cite{hydon2014difference}).
We also assume that all functions are smooth in their continuous arguments.

The total space is mapped to itself by horizontal translations
\begin{equation*}
\mathrm{T}_{\bf{I}}\colon\ \bbz^{m} \times \bbr^{q} \rightarrow
\bbz^{m} \times \bbr^{q} ,
\qquad \qquad
\mathrm{T}_{\bf{I}}\colon\ ( {\bf{n}},{\bf{u}} ) \mapsto ( {\bf{n}} + {\bf{I}},{\bf{u}} ),
\end{equation*}
where ${\bf I} \in \bbz^{m}$ is a multi-index. Over each $\mbn$, one can construct the prolongation space $P(\bbr^{q})$, which is an infinite-dimensional Cartesian product space with coordinates $\uJalpha\in\bbr$, where $\mbJ\in\bbz^m$. Let $P_\mbn(\bbr^{q})$ denote the prolongation space over a given $\mbn$; this is a continuous fibre over the fixed base point $\mbn$. Every graph on the total space defines a point in $P_\mbn(\bbr^{q})$, with the coordinate $\uJalpha$ taking the value of $ u^{\alpha} $ given by the graph at $ {\bf{n}} + {\bf{J}} $.

Horizontal translation extends naturally to the total prolongation space $ \bbz^{m} \times P ( \bbr^{q} ) $ as follows:
\begin{equation*}
\mathrm{T}_{\bf{I}}\colon\ \bigl( {\bf{n}} , \bigl( u_{{\bf{J}}}^{\alpha} \bigr)\bigr)
\mapsto
\bigl( {\bf{n}} + {\bf{I}} , \bigl( u_{{\bf{J}}}^{\alpha} \bigr)\bigr).
\end{equation*}
Suppose that $f$ is a function on $\bbz^m\times P (\bbr^{q} )$; its restriction to $P_{\bf n} (\bbr^{q} )$ is denoted by
\begin{equation*}
f_{\bf{n}} \bigl( \bigl( u_{{\bf{J}}}^{\alpha} \bigr) \bigr) = f \bigl( {\bf{n}} ,\bigl( u_{{\bf{J}}}^{\alpha} \bigr) \bigr).
\end{equation*}
The pullback $\mathrm{T}_{\bf I}^{\ast}$ of $f_{\bf{n} +{\bf I}} \bigl( \bigl( u_{{\bf{J}}}^{\alpha} \bigr) \bigr)$ to $P_{\bf n}(\bbr^{q})$ is
\begin{equation*}
\mathrm{T}_{\bf I}^{\ast} f_{{\bf n}+{\bf I}} \bigl( \bigl( u_{\bf J}^{\alpha} \bigr) \bigr) = f\bigl({\bf n}+{\bf I}, \bigl(u_{{\bf J}+{\bf I}}^{\alpha}\bigr)\bigr).
\end{equation*}
This can be represented on the prolongation space over $\mbn$ as a mapping $\s_\mbI$, called the \textit{shift} by~$\mbI$, which acts on smooth functions $f\in C^\infty(P_{{\bf{n}}} ( \bbr^{q} ))$ as follows:
\begin{equation*}
\s_{\bf{I}} f\bigl( {\bf{n}} , \bigl( u_{\bf{J}}^{\alpha}\bigr)\bigr) = f\bigl( {\bf{n}} + {\bf{I}} ,\bigl( u_{{\bf{J}} + {\bf{I}}}^{\alpha}\bigr)\bigr).
\end{equation*}
To summarize, each shift operator represents the action of the pullback on functions as follows: $ \s_{{\bf I}} f_{{\bf{n}}} := \mathrm{T}_{{\bf{I}}}^{\ast} f_{{{\bf{n}} + {\bf{I}}}}$. Although $\s_\mbI$ represents a translation, it does not change the fibre.
By using the shift operators, one can represent a given \pde as a variety on $P_{{\bf{n}}} ( \bbr^{q} )$. We do this from here on, using the Einstein summation convention to denote sums over all variables other than~$ \bf n $, as far as possible. For simplicity, we omit the multi-index subscript $\phantom{}_{\bf 0}$ on variables, except where this may cause confusion. In particular, $u^\alpha$ denotes $u^\alpha_{\bf 0}$ henceforth.

\begin{Remark}\label{restrict}
For a given \pdee, one can restrict attention to a finite-dimensional subspace of the prolongation space, provided that this includes all $\uJalpha$ for which $\mbJ$ lies within the stencil of the \pde with respect to $u^\alpha$ (see~\cite{hydon2014difference}), together with any other relevant $\uJalpha$ for the problem being considered. However, for generality, we use the full prolongation space over~$\mbn$.
\end{Remark}

\subsection{The difference variational calculus} \label{Section: The difference variational calculus}

This section summarizes the difference variational calculus from the formal viewpoint introduced by Kupershmidt \cite{kupershmidt1985discrete}. It closely resembles the formal differential variational calculus described in~\cite{olver2000applications}. Summation by parts replaces integration by parts, and a difference version of the divergence is used, which we now describe.

Each shift operator $\s_{{\bf J}}$, where $\mbJ=\bigl(j^1,\dots,j^m\bigr)$, may be written as a product of unit shift operators, $\s_i:=\s_{{\bf 1}_{i}}$, and their inverses. Here ${\bf 1}_{i}$ is the multi-index whose only non-zero entry is the $i^\mathrm{th}$ one, which is $1$. Thus $\s_i$ is the forward shift in the $n^i$-direction. By the composition rule for translations, $\s_i\s_j=\s_j\s_i$ and $\s_{\bf J}=\s_{1}^{j_{1}} \cdots \s_{m}^{j_{m}}$; consequently, $(\s_\mbJ)^{-1}=\s_{-\mbJ}$. The identity operator, $\id:=\s_{\bf 0}$, maps every function to itself, and the forward difference operator in the direction $n^{k} $ is
$\D_{n^{k}} = \s_{k}- \id$.
A \textit{difference divergence} is an expression of the form $\operatorname{Div} (F) = \D_{n^{k}} F^{k}$ for some $F:= \bigl(F^{1}, \dots, F^{m}\bigr)$. It is straightforward to write a given expression of the form $(\s_\mbJ -\id)f$, where $f$ is a function, as a difference divergence; however, the resulting functions $F^k$ are not unique if $m>1$ and may be messy.

\begin{Definition}
A \textit{conservation law} of a given system of \pdes is a difference divergence expression, $\mathcal{C}=\operatorname{Div} (F)$, such that $\mathcal{C}=0$ on all solutions of the system.
\end{Definition}

A linear difference operator on $P_{{\bf{n}}} (\bbr^{q})$ is an operator of the form $\mcH=h^\mbJ\s_\mbJ$, where each $h^\mbJ$ is a function. The formal adjoint of $\mcH$ is the operator $\mcH^\dagger$ defined by
\[
f\mcH g -\bigl(\mcH^\dagger f\bigr)g\in \mathrm{im(Div)}
\]
for all functions $f$, $g$. Explicitly, $\mcH^\dagger f=\s_{-\mbJ}\bigl(h^\mbJ f\bigr)$, because
\[
fh^\mbJ\s_\mbJ g-\bigl(\s_{-\mbJ}\bigl(h^\mbJ f\bigr)\bigr)g=(\s_\mbJ -\id)\bigl\{\bigl(\s_{-\mbJ}\bigl(h^\mbJ f\bigr)\bigr)g\bigr\}.
\]
A special case is the very useful \textit{summation by parts} formula
\begin{equation*} 
	f ( \s_{\bf J} g ) =
 ( \s_{-\bf J} f )g + ( \s_{\bf J} - \id ) \{ ( \s_{-\bf J} f )g \}.
\end{equation*}

The basic variational problem is to find the extrema of a given functional
\begin{equation*} 
\mathcal{L} [ {\bf u}] = \sum_{\bf n} \Lrm ({\bf n}, [{\bf u}] ),
\end{equation*}
where $[{\bf u}]$ represents finitely many shifts of the dependent variables.
Extrema are found by requiring that
\begin{equation*}
\bigg\lbrace \frac{\de}{\de \epsilon} \mathcal{L} [ {\bf u} + \epsilon {\bf w} ] \bigg\rbrace \bigg\vert_{\epsilon=0} =0,
\end{equation*}
for all $ {\bf w}\colon \mathbb{Z}^{m} \rightarrow \mathbb{R}^{q} $ that vanish sufficiently rapidly as any independent variable approaches infinity.
Using summation by parts,
\begin{align*}
\frac{\de}{\de \epsilon} \bigg \vert_{\epsilon=0} \Lrm (\mbn,[{\bf u}+\epsilon{\bf w}] ) = \left({\s}_{{\bf J}} w^{\alpha} \frac{\p \Lrm}{\p u_{{\bf J}}^{\alpha}} \right)= w^{\alpha} \E_{u^{\alpha}} ( \Lrm ) + \operatorname{Div} ( A_{\bf u} ( {\bf n},[\mbu],[{\bf w}] ) ),
\end{align*}
where
\begin{equation*}
	\E_{u^{\alpha}} = \s_{-\bf J} \frac{\p}{\p \uJalpha}
\end{equation*}
is the difference Euler--Lagrange operator with respect to $u^{\alpha}$, and
\begin{equation*}
	\operatorname{Div} ( A_{\bf u} ( {\bf n},[\mbu],[{\bf w}] ) ) = \sum_{\bf J} ( \s_{\bf J}-\id )\left(w^{\alpha} {\s}_{-{\bf J}} \frac{\p \Lrm}{\p u_{{\bf J}}^{\alpha}} \right).
\end{equation*}

As $\mathbf{w}$ is arbitrary and the sum of $\operatorname{Div} ( A_{\bf u} ( {\bf n},[\mbu],[{\bf w}] ) )$ over $\mbn$ is zero (by Stokes' theorem), the extrema satisfy the following system of Euler--Lagrange (difference) equations,
\begin{equation*}
\E_{u^{\alpha}} (\Lrm) = {\s}_{{-\bf J}}\left( \frac{\p \Lrm }{\p u_{{\bf J}}^{\alpha}} \right)=0.
\end{equation*}

\begin{Remark}\label{remt}
If the dependent variables are regarded as depending smoothly on a continuous parameter $t$ as well as $\mbn$, the same result is achieved by using
\begin{equation*}
\frac{\de}{\de t} \bigg \vert_{(u^{\alpha})'=w^{\alpha}} \Lrm [{\bf u}] =0 ,
\end{equation*}
where $(u^{\alpha})'=\de u^{\alpha}/\de t$.
This approach is used later to derive the Lie group invariant version of the Euler--Lagrange equations.
\end{Remark}

\subsection{Variational point symmetries and Noether's theorem}

We now outline some relevant facts about Lie point symmetries, with application to variational calculus (see \cite{hydon2014difference,olver2000applications} for further details). Let $G$ be an $R$-dimensional Lie group parametrized by~$\boldsymbol{\varepsilon} = \bigl( \varepsilon^{1},\dots,\varepsilon^{R} \bigr)\in\bbr^R$ in some neighbourhood of the identity, $e$. For now, we restrict attention to such a neighbourhood, in which the elements of $G$ are $ \Gamma ( \boldsymbol{\varepsilon} ) $, where $\Gamma$ depends smoothly on~$\boldsymbol{\varepsilon}$, with $ \Gamma ( {\bf 0} ) = e $.
Locally, the left action of $G$ on the coordinates $ {\bf u} = \bigl( u^{1},\dots,u^{q} \bigr) $ is denoted by~$ \widehat{\bf u} = \Gamma ( \boldsymbol{\varepsilon} ) \cdot {\bf u} $. The $R$-dimensional Lie algebra $ \mathcal{X} $ of infinitesimal generators has a~basis
\begin{equation*} 
	{\bf v}_{r} = Q_{r}^{\alpha} ( {\bf n}, {\bf u} ) \p_{u^{\alpha}}, \qquad r=1,\dots, R,\qquad \text{where} \quad Q_{r}^{\alpha} = \frac{\p \widehat{u}^{\alpha}}{\p \varepsilon^{r}} \bigg \vert_{\boldsymbol{\varepsilon} = {\bf 0}},
\end{equation*}
so every infinitesimal generator of a one-parameter (local) Lie subgroup of point transformations is of the form ${\bf v} = Q^{\alpha} ( {\bf n}, {\bf u} ) \p_{u^{\alpha}}$,
where $Q^{\alpha}=c^rQ_{r}^{\alpha}$ for some real constants $c^r$.
The $q$-tuple $\mathbf{Q}=\bigl(Q^1,\dots, Q^R\bigr)$ is the \textit{characteristic} of the Lie subgroup whose infinitesimal generator is $\mathbf{v}$.

Each infinitesimal generator is a tangent vector field on the total space. It is represented on the prolongation space $P_{\bf n} (\bbr^{q} ) $ by
the prolonged vector field
\begin{equation*}
	\operatorname{pr} {\bf v} = (\s_{\bf J} Q^{\alpha}) \frac{\p}{\p u_{\bf J}^{\alpha}}.
\end{equation*}
From here on, we refer to $\operatorname{pr} {\bf v}$ simply as ${\bf v}$, because it will always be clear whether the generator is acting on the total space or the prolongation space over $\mbn$. Note that $\mbn$ is invariant under the Lie group action, as are the shift operators $\s_\mbJ$.

Denote the left action of a general group element $g\in G$ (not necessarily in the neighbourhood of the identity) on the total space by $ \widetilde{\bf u} = g \cdot {\bf u}$ and define
\begin{equation*}
	\widetilde{\bf v}_{r} = Q_{r}^{\alpha} ({\bf n},\widetilde{\bf u} ) \p_{\widetilde{u}^{\alpha}}, \qquad r=1,\dots,R.
\end{equation*}
The adjoint representation of $g$ on $ \mathcal{X} $ can be expressed as a matrix, $\mathcal{A}{\rm d}(g)= ( a_{r}^{s}(g) ) $, whose components are determined from the following relations (see \cite{mansfield2019movinga}):
\begin{equation} \label{Eq: Adjoint identity}
	{\bf v}_{r} = a_{r}^{s}(g)\widetilde{\bf v}_{s}, \qquad r = 1,\dots,R.
\end{equation}
By regarding the infinitesimal generators as differential operators and applying the left-hand side of the identity (\ref{Eq: Adjoint identity}) to each $ \widetilde{u}^{\alpha} $ in turn, one obtains
\begin{equation} \label{Eq: Adjoint identity 2}
	\left(\frac{\p \widetilde{u}^{\alpha}}{\p u^{\beta}}\right) Q_{r}^{\beta} ({\bf n},{\bf u} )={\bf v}_{r} (\widetilde{u}^{\alpha} )= Q_{s}^{\alpha} ({\bf n},\widetilde{\bf u} ) a_{r}^{s}(g),
\end{equation}
where the matrix $ \bigl(\p \widetilde{u}^{\alpha}/ \p u^{\beta}\bigr)$ is the Jacobian matrix of the transformation $g\colon\mbu\rightarrow\widetilde{\mbu}$.
Prolonging this result to each $\uJalpha$ gives the identities
\begin{equation*}
	\left( \frac{\p \widetilde{u}_{\bf J}^{\alpha}}{\p u_{\bf J}^{\beta}}\right) Q_{r}^{\beta} ( {\bf n + J}, {\bf u_{J}} ) = Q_{s}^{\alpha} ( {\bf n} + {\bf J}, \widetilde{\bf u}_{\mbJ} ) a_{r}^{s}(g).
\end{equation*}

\begin{Definition}
	The point transformations generated by $\mbv$ are \textit{variational symmetries} of the Lagrangian $\Lrm(\mbn,[\mbu])$ if there exist functions $B^{i} ({\bf n},[{\bf u}] )$ such that
	\begin{equation}\label{Eq: Variational symmetries}
		\mathbf{v}(\Lrm):=(\s_{\bf J} Q^{\alpha}) \frac{\p \Lrm}{\p u_{\bf J}^{\alpha}} = \D_{n^{i}} B^{i}.
	\end{equation}
\end{Definition}

The Lagrangian is invariant under the symmetries generated by $\mathbf{v}$ if $B^{i}=0 $ for all $i$.
Summing~\eqref{Eq: Variational symmetries} by parts leads to Noether's theorem for P$\Delta$Es (see \cite{hydon2014difference}).

\begin{Theorem}[difference Noether's theorem] \label{Thm: Noether's theorem}
	Suppose that a Lagrangian $\Lrm$ has a variational symmetry with characteristic $\bf{Q} \neq {\bf 0}$. Then the system of Euler--Lagrange equations has the following conservation law:
	\begin{equation} \label{Eq: Difference Noether's}
-\D_{n^{i}} B^{i}+\sum_{\bf J} ( \s_{\bf J}-\id )\left(Q^{\alpha} {\s}_{-{\bf J}} \frac{\p \Lrm}{\p u_{{\bf J}}^{\alpha}} \right) =0.		
	\end{equation}
\end{Theorem}

From here on, we restrict attention to Lie groups of variational point symmetries that leave the Lagrangian invariant, so $B^i=0,\ i=1,\dots,m$.

\subsection{Moving frames} \label{Section: Moving frames}

We now outline some basics of moving frames on an arbitrary manifold (see \cite{fels1998moving,fels1999moving,mansfield2010practical} for further details).
Let $M$ be a smooth manifold. Suppose that a Lie group, $G$, of point transformations has a smooth left action on $M$ that is free and regular\footnote{The action is free if the only group element $g \in G$ that fixes every point in the neighbourhood is the identity. The action is regular if the orbits form a regular foliation.} in a neighbourhood $\mathcal{U}\subset M$ of a~point~$z\in M$.
Freeness and regularity are necessary and sufficient to guarantee the existence of a cross-section~$\mcK$ that is transverse to the group orbits that foliate $\mathcal{U}$, as shown in Figure \ref{Fig: Moving frame image}. Moreover, each orbit intersects $\mcK$ at a unique point.

If a group action is not free and regular, it can be made so by replacing $M$ by a Cartesian product space $M^N$ of sufficiently high dimension and using the induced product action (see Boutin \cite{boutin2002orbit}). Henceforth, we assume that $M$ is of sufficiently high dimension that the action is free and regular.

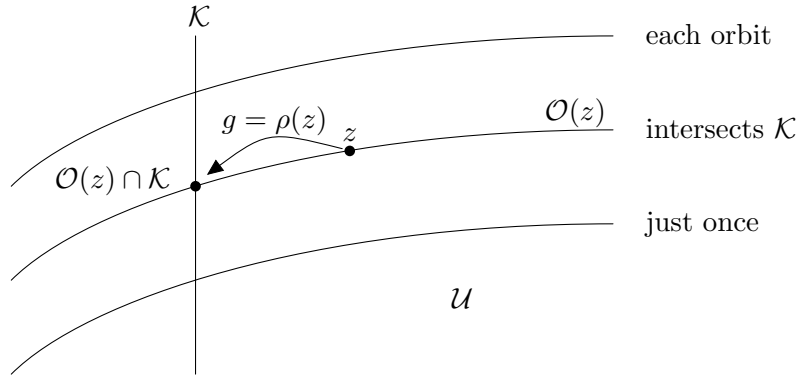
\begin{figure}[t]
\begin{center}
\begin{tikzpicture}[line cap=round,line join=round,>=triangle 45]
\draw (0,-0.25) .. controls (1,0.75) and (4,1.75).. (8,1.75);
\draw (0,1) .. controls (1,2) and (4,3).. (8,3);
\draw (0,2.25) .. controls (1,3.25) and (4,4.25).. (8,4.25);
\draw (2.45,-0.25) -- (2.45,4.25);
\draw (8.3,4.27) node[anchor=west] {each orbit};
\draw (8.3,3.03) node[anchor=west] {intersects $\mathcal{K}$};
\draw (8.3,1.74) node[anchor=west] {just once};
\draw (2.5,4.25) node[anchor=south] {$\mathcal{K}$};
\draw (7.5,2.9) node[anchor=south] {$\mathcal{O}(z)$};
\fill (2.45,2.25) circle[radius=2pt];
\draw (1.35,2) node[anchor=south] {$\mathcal{O}(z) \cap \mathcal{K}$};
\fill (4.5,2.72) circle[radius=2pt];
\draw (4.5,2.72) node[anchor=south] {$z$};
\draw (3.5,2.8) node[anchor=south] {$g=\rho (z)$};
\draw[->] (4.4,2.75) .. controls (3.4,3) .. (2.6,2.4);
\draw (6,0.5) node[anchor=south] {$\mathcal{U}$};
\end{tikzpicture}
\end{center}
\caption{Moving frame defined by a cross-section; $\mathcal{O}(z)$ denotes the group orbit through $z$.}\label{Fig: Moving frame image}
\end{figure}

\begin{Definition}[moving frame]
	Given a smooth Lie group action $G \times M \rightarrow M $, a moving frame is a smooth equivariant map $ \rho\colon \mathcal{U} \subset M \rightarrow G $. Here $ \mathcal{U} $ is called the domain of the frame.
\end{Definition}

Given a left action, $g \cdot z$, a left equivariant map satisfies $ \rho ( g \cdot z ) = g \rho (z)$ and a right equivariant map satisfies $ \rho ( g \cdot z ) = \rho (z) g^{-1}$.
The frame is called left or right accordingly.
The inverse of a~right frame is a~left frame.

To find a right frame for an $R$-dimensional Lie group, $G$, write the cross-section $ \mathcal{K} $ as a system of equations $ \psi_{r} (z) = 0$, $r = 1,\dots,R $.
Then solve the \textit{normalization equations},
\begin{equation} \label{Eq: Normalization equations general}
	\psi_{r} (g \cdot z) = 0, \qquad r = 1,\dots,R \text{,}
\end{equation}
to obtain the unique group element $g = \rho (z) $ that maps $ z $ to the intersection of its group orbit with $\mathcal{K}$ (see Figure \ref{Fig: Moving frame image}).
Both $ \rho (g \cdot z) $ and $ \rho (z) g^{-1} $ satisfy the equation $\psi_{r} ( \rho (g \cdot z) \cdot (g \cdot z) )=0$,
so by uniqueness, the solution is a right frame.

An important part of the method of moving frames is to choose a cross-section $\mathcal{K}$ that makes computations as simple as possible. It is usually easiest to choose a cross-section on which $R$ of the coordinates on $\mathcal{U}$ are constant. The normalization equations are then expressed as
\begin{equation*}
	g \cdot z_{1} = c_{1}, \qquad g \cdot z_{2} = c_{2}, \qquad \dots, \qquad g \cdot z_{R} = c_{R},
\end{equation*}
where $z_{r}$ are coordinates and $c_{r}$ are fixed constants. We will use this approach in our examples.

Given a left action $G \times M \rightarrow M$ and a right frame $ \rho $, let
\begin{equation*}
	\iota(z) = \rho (z) \cdot z= g \cdot z \vert_{g=\rho(z)}.
\end{equation*}
The components of $\iota(z)$ are invariant under the Lie group action, because
\begin{equation*}
		\iota (g \cdot z) = \rho (g \cdot z) \cdot (g \cdot z) = \rho (z) \cdot g^{-1} g \cdot z = \rho(z) \cdot z = \iota (z).
\end{equation*}
The non-constant components of $\iota(z)$ are called the \textit{normalized invariants}. Note that $\rho(\rho(z)\cdot z)$ is the identity element of $G$, so $\iota(\iota(z))=\iota(z)$. In other words, the operator $\iota$ projects $z$ to its invariant components. This operator extends to functions $f(z)$, as follows: $\iota (f (z) )= f ( \iota (z) )$.
Note that $\iota(\iota(f(z)))=\iota(f(z))$, so $\iota$ projects out the invariant component of $f(z)$; hence, $\iota$ is called the \textit{invariantization operator}. Indeed, if $F(z)$ is any invariant, the following \textit{replacement rule} applies:
\begin{equation}\label{Eq: Replacement rule}
	F(z)=F(\iota(z)).
\end{equation}
Consequently, the set of all invariants is generated by the normalized invariants $\iota(z)$.

\subsection{Discrete moving frames}

The discrete moving frame developed by Mansfield, Mar\'{i} Beffa and Wang \cite{mansfield2013discrete} and Mar\'{i} Beffa and Mansfield \cite{maribeffa2018discrete} can be thought of as a moving frame adapted to discrete base points.
The Lie group action on $M$ is extended to the diagonal (left) action on the Cartesian product manifold~$\mathcal{M} = M^{N} $:
\begin{equation*}
	g \cdot (z_{1},z_{2},\dots,z_{N} ) \mapsto ( g \cdot z_{1}, g \cdot z_{2},\dots, g \cdot z_{N} ).
\end{equation*}
No assumptions are made about any relationship between the elements $z_{1},\dots,z_{N}$.

\begin{Definition}[discrete moving frames]
	Let $ G^{N} $ denote the Cartesian product of $N$ copies of the group $ G$.
	A map
	\begin{equation*}
		\rho\colon\ M^{N} \rightarrow G^{N}, \qquad \rho(z) = (\rho_{1} (z),\dots, \rho_{N}(z) )
	\end{equation*}
	is a right discrete moving frame if
	\begin{equation*}
		\rho_{k} (g \cdot z) = \rho_{k}(z)g^{-1}, \qquad k= 1,\dots,N,
	\end{equation*}
	and a left discrete moving frame if
	\begin{equation*}
		\rho_{k} (g \cdot z) = g\rho_{k}(z), \qquad k= 1,\dots,N.
	\end{equation*}
The frame is right (resp.\ left) equivariant under the action of the Lie group.
\end{Definition}

A set of normalization equations yields a corresponding right discrete moving frame. The component $\rho_{k}$ is the unique element of $G$ that takes $z$ to the cross-section $ \mathcal{K}_{k} $.
The sequence of moving frames with a nontrivial intersection of domains $ ( \rho_{k} ) $ which makes up the discrete moving frame is, locally, uniquely determined by the cross-section $\mathcal{K}= ( \mathcal{K}_{1},\dots,\mathcal{K}_{N} )$ to the group orbit through $z$ (see \cite{mansfield2013discrete} for more details). Again, no assumptions are made about any relationship between the components of the cross-section.

The invariants of the right (discrete) frame are $I_{k,j}:=\rho_{k} (z) \cdot z_{j}$. If $M$ is $q$-dimensional, each~$z_{j}$ has components $z_{j}^{1},\dots,z_{j}^{q}$. So the components of $I_{k,j}$ are
\begin{equation*}
	I_{k,j}^{\alpha}:= \rho_{k}(z) \cdot z_{j}^{\alpha}, \qquad \alpha = 1,\dots,q.
\end{equation*}

The discrete moving frame applies to a wide variety of discrete domains.
We now show how it can be adapted to the difference prolongation space for \pdese, yielding the difference moving frame.

\section{Difference moving frames}\label{3sec}

In view of Remark \ref{restrict}, one can restrict attention to a finite prolongation space $\mathcal{M}$ of arbitrarily high dimension. This enables one to treat a difference moving frame as a particular type of discrete moving frame, so that the key definitions and theorems for discrete moving frames in~\cite{mansfield2013discrete,mansfield2019movinga,maribeffa2018discrete} apply. The main distinction is that for the difference frame, there is a relation between the cross-sections and frames on the different fibres: they must be consistent with the pullback to any particular fibre.

\subsection{Difference frames and invariants}

Let $ \mathcal{K} $ and $ \rho ( [{\bf u} ] ) $ denote the cross-section and frame on~$ {\bf n} $, respectively.
The cross-section on~$\bf n$, denoted $ \mathcal{K} $, is replicated for all the other base points $ {\bf n} + {\bf J} $ if and only if the cross-section over~$ {\bf n} + {\bf J} $ is represented on $\mathcal{M}$ by $ {\s_{\bf J}} \mathcal{K}$. Consequently, $\rho_\mbJ([\mbu])=\s_\mbJ\rho([\mbu])$; this constraint can be extended to the infinite-dimensional prolongation space $P_\mbn(\bbr^{q})$. From here on we use $\rho$ (resp.~$\rho_\mbJ$) as shorthand for $\rho_{\mathbf{0}}([\mbu])$ (resp.\ $\rho_\mbJ([\mbu])$).

\begin{Definition}
A difference moving frame is a discrete moving frame such that $\mathcal{M}$ is a prolongation space and the cross-section over $ {\bf n} + {\bf J} $ is represented on $\mathcal{M}$ by $ {\s_{\bf J}} \mathcal{K}$.
\end{Definition}
The normalized invariants for difference moving frames are
\begin{equation*}
I_{\bf K,\bf J}:= \rho_{\bf K} \cdot {\bf u}_{\bf J}= (\s_{\bf K}\rho ) \cdot (\s_{\bf J} {\bf u} ).
\end{equation*}
By definition, $\s_{i}I_{\bf K,\bf J}= I_{{\bf K}+{\bf 1}_{i},{\bf J}+{\bf 1}_{i}}$.
Hence, every invariant $I_{\bf K,\bf J}$ can be expressed as a shift of
\begin{equation}\label{Ionu}
	I_{\mbK-\mbJ,\bf 0} = (\s_{\mbK-\mbJ} \rho) \cdot {\bf u}.
\end{equation}

\begin{Definition}[discrete Maurer--Cartan invariants]
Given a right discrete moving frame $ \rho ( [{\bf u} ] ) $ (which commonly is expressed as a matrix), the right discrete \textit{Maurer--Cartan invariants} are the components of
\begin{equation*}
K_{(i)} = ( \s_{i} \rho ) \rho^{-1} = \iota ( \s_{i} \rho ), \qquad i = {1,\dots,m},
\end{equation*}
together with their shifts $\s_\mbJ K_{(i)}$. Invariance of $K_{(i)} $ follows from equivariance of the frame.
\end{Definition}

\begin{Definition}
A set of invariants is a \textit{generating set} for the algebra of functions of difference invariants if any difference invariant in the algebra can be written as a function of elements of the generating set and their shifts.
\end{Definition}

The invariantization of a multiply-shifted frame is obtained by concatenating the matrices of Maurer--Cartan invariants and their shifts. For instance,
\begin{equation*}
\iota ( \s_{i}\s_j \rho )= ( \s_{i}\s_{j}\rho )\rho^{-1}= ( \s_{j} K_{(i)} ) K_{(j)}.
\end{equation*}
Consequently, \eqref{Ionu} can be written as $I_{\mbK-\mbJ,{\bf 0}} = \bigl\{(\s_{\mbK-\mbJ} \rho)\rho^{-1} \bigr\}I_{\mathbf{0},\mathbf{0}}$,
where the term in braces is a concatenation of the Maurer--Cartan invariants and their shifts. This establishes the following result.

\begin{Proposition}\label{Prop: geninv}
Given a right difference moving frame $\rho ( [{\bf u} ] )$, the set of all invariants is generated by $I_{\bf{0},\bf{0}}=\rho\cdot\mbu$ and the set of components of $K_{(j)} = ( \s_{j} \rho ) \rho^{-1}$, $j=1,\dots,m $.
\end{Proposition}

\begin{Definition}[syzygy]\label{Def: Syzygy}
	A syzygy on a set of invariants is a relation between the invariants that expresses functional dependency.
\end{Definition}

In other words, a syzygy on a set of invariants is a function of the invariants that becomes an identity when it is expressed in terms of the underlying variables $[\mbu]$. In general, there are syzygies between the invariants in Proposition \ref{Prop: geninv}, which lead to useful recurrence relations. Such relations enable all of these invariants to be expressed in terms of a small\footnote{The number of generating invariants depends on the number of dependent and independent variables, as well as on the group $G$ and the normalization that is used. Finding a relation between these quantities is an open problem.} generating set of invariants, $\kappa^\beta$, and their shifts, $\kappa^\beta_\mbJ:=\s_\mbJ\kappa^\beta$.

\begin{Example}\label{Example: Normalization}
	As a running example to illustrate the theory, we use the Lagrangian
	\begin{equation} \label{Eq: Lagrangian}
		\Lrm
		=
		\ln \biggl\vert \frac{ u_{1,0}-u_{0,1} }{ u_{1,1}-u_{0,0} } \biggr\vert,
	\end{equation}
	where $\mbn=\bigl(n^1,n^2\bigr)$ and $u_{i,j}$ represents the value of $u$ at $\bigl(n^1+i,n^2+j\bigr)$. The Euler--Lagrange equation is
	\begin{equation}\label{TodaEL}
		\E_{u}(\Lrm) = \frac{1}{u_{1,1}-u_{0,0}} - \frac{1}{u_{-1,1}-u_{0,0}} - \frac{1}{u_{1,-1}-u_{0,0}} + \frac{1}{u_{-1,-1}-u_{0,0}}=0,
	\end{equation}
	which is a Toda-type equation that is satisfied by all solutions of the autonomous dpKdV equation and the cross-ratio equation (see~\cite{hydon2014difference} for details). This equation is partitioned into two independent components, with $n^1+n^2$ being either odd or even.
	
	The Lagrangian \eqref{Eq: Lagrangian} has a six-parameter Lie group of variational symmetries, whose infinitesimal generators are linear combinations of
	\begin{align*}
		&{\bf v}_{1} = \p_{u_{0,0}}, \qquad {\bf v}_{2} = u_{0,0} \p_{u_{0,0}},
		\qquad {\bf v}_{3} = u_{0,0}^{2} \p_{u_{0,0}}, \notag\\
		&{\bf v}_{4} = (-1)^{n^1+n^2} \p_{u_{0,0}}, \qquad {\bf v}_{5} = (-1)^{n^1+n^2}u_{0,0} \p_{u_{0,0}},
		\qquad {\bf v}_{6} = (-1)^{n^1+n^2} u_{0,0}^{2} \p_{u_{0,0}},
	\end{align*}
	where we use $\p_z$ as shorthand for $\p/\p z$ from here on. The corresponding characteristics are
	\begin{gather*}
	Q_1=1,\qquad Q_2=u_{0,0},\qquad Q_3=u_{0,0}^2,\qquad Q_4=(-1)^{n^1+n^2},\\ Q_5=(-1)^{n^1+n^2}u_{0,0},\qquad Q_6=(-1)^{n^1+n^2}u_{0,0}^2,
	\end{gather*}
	and the prolonged generators are
	\begin{equation*}
		\operatorname{pr} {\bf v}_{r} = \bigl(\s_1^i\s_2^j Q_{r}\bigr) \p_{u_{i,j}} .
	\end{equation*}
Only the symmetries generated by linear combinations of $\mbv_1$, $\mbv_2$ and $\mbv_4$ leave the Lagrangian invariant. (The other generators produce divergence terms, although these can be absorbed without changing the Euler--Lagrange equation, by adding a divergence to the Lagrangian.)

Consider the two-parameter Lie group action generated by ${\bf v}_{1} $ and ${\bf v}_{2} $; this is
	\begin{equation*} 
		g\colon\ u_{i,j} \mapsto \widetilde{u}_{i,j} = b u_{i,j} + a,
	\end{equation*}
	which is defined for every $a \in \bbr$ and $b \in \bbr^{+}$. For the half-space $\mathcal{U} = \{\mathcal{M}\colon u_{1,1} > u_{0,0} \}$, we choose the normalization equations\footnote{For the other half-space, with $ u_{1,1} < u_{0,0}$, an appropriate normalization is $\widetilde{u}_{0,0} = 0 $ and $\widetilde{u}_{1,1} = -1 $.} (\ref{Eq: Normalization equations general}) to be
	\begin{equation*}
		\widetilde{u}_{0,0}=0, \qquad \widetilde{u}_{1,1} = 1.
	\end{equation*}
Then the frame $\rho$ is the group element with the parameters
\begin{equation*}
a=\frac{-u_{0,0}}{u_{1,1}-u_{0,0}} , \qquad b=\frac{1}{u_{1,1}-u_{0,0}} .
\end{equation*}
For this frame, the invariantization of $u_{i,j} $ is
\begin{equation*}
	\iota (u_{i,j}) = (\widetilde{u}_{i,j})\big|_{g=\rho}= \frac{u_{i,j}-u_{0,0}}{u_{1,1}-u_{0,0}} .
\end{equation*}
The basic Maurer--Cartan invariants are
\begin{equation*}
	\kappa=\iota (u_{1,0} )=\frac{u_{1,0}-u_{0,0}}{u_{1,1}-u_{0,0}} , \qquad \lambda=\iota (u_{0,1} )=\frac{u_{0,1}-u_{0,0}}{u_{1,1}-u_{0,0}} .
\end{equation*}
These two invariants generate all invariants $\iota ( u_{i,j} )$.

As an example, we show how to find $\iota ( u_{2,1} )$ in terms of these generating invariants and their shifts.
First shift $\lambda$ to involve $u_{2,1}$ and lower shifts of $\kappa$ and $\lambda$:
\begin{equation*}
	\lambda_{1,0}=\frac{u_{1,1}-u_{1,0}}{u_{2,1}-u_{1,0}} .
\end{equation*}
As shifts of invariants are invariant, the replacement rule yields
\begin{equation*}
	\lambda_{1,0}=\frac{\iota (u_{1,1})-\iota (u_{1,0})}{\iota (u_{2,1})-\iota (u_{1,0})}=\frac{1-\kappa}{\iota (u_{2,1})-\kappa} .
\end{equation*}
Rearranging,
\begin{equation*}
	\iota (u_{2,1}) =\kappa+\frac{1-\kappa}{\lambda_{1,0}} .
\end{equation*}
Similarly, the replacement rule gives
\[
\s_{1} \iota( u_{i,j})=\frac{u_{i+1,j}-u_{1,0}}{u_{2,1}-u_{1,0}}=\frac{\iota(u_{i+1,j})-\kappa}{\iota(u_{2,1})-\kappa} ,
\]
which leads to the first of two general recurrence relations:
\begin{align}
&\iota(u_{i+1,j})=\kappa+\left(\frac{1-\kappa}{\lambda_{1,0}} \right) \s_{1} \iota( u_{i,j}),\label{rr1}\\
&\iota (u_{i,j+1} )=\lambda+\left(\frac{1-\lambda}{\kappa_{0,1}} \right) \s_{2} \iota( u_{i,j}).\label{rr2}
\end{align}
Note that the invariantization operator $\iota$ does not commute with the shift operators.

The invariant $\iota(u_{2,2})$ can be calculated from the identity $\iota(u_{1,1})=1$ in two ways: either use~\eqref{rr1} first, then \eqref{rr2}, or vice versa. This leads to the following syzygy between the generating invariants and their shifts:
\begin{equation}\label{fundsyz}
\frac{(\lambda-1)(\kappa_{0,1}-1)}{\kappa_{0,1}\lambda_{1,1}} = \frac{(\kappa-1)(\lambda_{1,0}-1)}{\lambda_{1,0}\kappa_{1,1}}.
\end{equation}
The syzygy \eqref{fundsyz} arises because two dependent variables ($\kappa,\lambda$) are used instead of one ($u$).
\end{Example}

\subsection{Differential invariants and syzygies}

From Remark \ref{remt}, one can derive the Euler--Lagrange equations by regarding the variables $\uJalpha$ as depending smoothly on a continuous parameter~$t$. The same is true in the context of invariants, if one stipulates that
\begin{itemize}\itemsep=0pt
	\item $t$ is invariant under the Lie group action;
	\item every shift commutes with differentiation with respect to $t$ (as is required for differential-difference equations~\cite{peng2022transformations}).
\end{itemize}
This approach was used in Mansfield et al.~\cite{mansfield2019movinga} to invariantize the Euler--Lagrange equations for~O$\Delta$Es. We now extend it to P$\Delta$Es.

As $t$ is invariant, the Lie group action (for point transformations) extends to the first-order jet space of $\mathcal{M}$ as follows:
\begin{equation*}
g \cdot \frac{\de \uJalpha}{\de t} = \frac{\de (g \cdot \uJalpha)}{\de t}= \frac{\p (g \cdot u_{\bf J}^{\alpha})}{\p u_{\bf J}^{\delta}} \frac{\de u_\mbJ^{\delta}}{\de t} .
\end{equation*}
As the action is free and regular on $\mathcal{M}$, it will remain so on the jet space and we may use the same frame to find the first-order differential invariants. Let
\begin{equation}\label{Gonder1}
\sigma^\alpha:=\iota\left(\frac{\de u^{\alpha}}{\de t}\right)= \frac{\p (g \cdot u^{\alpha})}{\p u^{\delta}}\bigg\vert_{g=\rho} \frac{\de u^{\delta}}{\de t} ,\qquad \alpha=1,\dots,q.
\end{equation}
The Jacobian matrix of any Lie group transformation on $\mathcal{M}$ is necessarily non-singular, so \eqref{Gonder1} can be inverted, as follows:
\begin{equation*}
\frac{\de u^{\delta}}{\de t}=\theta^\delta_\alpha([\mbu]) \sigma^\alpha.
\end{equation*}
Here the coefficients $\theta^\delta_\alpha([\mbu])$ are the components of the inverse of the Jacobian matrix, evaluated on the moving frame. Consequently, for each $\mbJ$,
\begin{equation}\label{diffinv}
\iota\left(\frac{\de u_{\bf J}^{\delta}}{\de t}\right)=\iota\bigl(\s_\mbJ \bigl(\theta^\delta_\alpha([\mbu]) \sigma^\alpha\bigr)\bigr)=\iota\bigl(\s_\mbJ \theta^\delta_\alpha([\mbu])\bigr)\s_\mbJ \sigma^\alpha.
\end{equation}
So all of the first-order differential invariants can be written in terms of the generating invariants~$\kappa^\beta$, the generating differential invariants $\sigma^\alpha$, and their shifts. In particular,
\begin{equation*}
	\sigma^\delta=\iota\left(\frac{\de u^{\delta}}{\de t}\right)=\iota\bigl(\theta^\delta_\alpha([\mbu])\bigr) \sigma^\alpha,
\end{equation*}
so the invariantization of the matrix $\bigl(\theta^\delta_\alpha([\mbu])\bigr)$ is the identity matrix.

To calculate the invariantized Euler--Lagrange equations, it is necessary to determine the differential syzygies
\begin{equation*}
\frac{\de \kappa^{\beta}}{\de t} =\iota\left(\frac{\p\kappa^\beta}{\p u^\alpha_\mbJ}\right) \iota\left(\frac{\de u_{\bf J}^{\alpha}}{\de t}\right)= \mathcal{H}_{\alpha}^{\beta}\sigma^{\alpha}.
\end{equation*}
The terms $\mathcal{H}_{\alpha}^{\beta} $ are linear difference operators whose coefficients are functions of $\kappa^{\beta}$ and their shifts.

\begin{Example}[{Example \ref{Example: Normalization} cont.}] \label{Example: Differential-difference syzygy}
We now find the differential invariants for the running example.
The action of the group on the derivative $u_{i,j}' =\de u_{i,j}/\de t $ is
\begin{align*}
g \cdot u_{i,j}' = \frac{\p (g \cdot u_{i,j})}{\p u_{i,j}} u_{i,j}' =b u_{i,j}'.
\end{align*}
Therefore, the first-order differential invariants are
\begin{equation*}
\iota ( u_{i,j}' ) = \frac{u_{i,j}'}{u_{1,1}-u_{0,0}} .
\end{equation*}
The generating differential invariant is
\begin{equation*}
\sigma = \frac{u_{0,0}'}{u_{1,1}-u_{0,0}} ,
\end{equation*}
so \eqref{diffinv} amounts to
\[
\iota(u_{i,j}')=\{\iota(u_{i+1,j+1})-\iota(u_{i,j})\}\s_1^i\s_2^j\sigma.
\]
For instance,
\[
\iota(u_{0,1}')=\{\iota(u_{1,2})-\iota(u_{0,1})\}\s_2\sigma=\frac{1-\lambda}{\kappa_{0,1}} \s_2\sigma.
\]
The derivatives of the generating invariants are
\begin{align}
&\frac{\de \kappa}{\de t}= \frac{u_{1,0}'-u_{0,0}'}{u_{1,1}-u_{0,0}} - \frac{(u_{1,0}-u_{0,0}) (u_{1,1}'-u_{0,0}' )}{(u_{1,1}-u_{0,0} )^{2}} , \nonumber \\
&\frac{\de \lambda}{\de t}= \frac{u_{0,1}'-u_{0,0}'}{u_{1,1}-u_{0,0}} - \frac{(u_{0,1}-u_{0,0}) (u_{1,1}'-u_{0,0}' )}{(u_{1,1}-u_{0,0} )^{2}} .\label{Eq: Derivatives of the generating invariants with respect to t}
\end{align}
Invariantizing \eqref{Eq: Derivatives of the generating invariants with respect to t}, using the replacement rule (\ref{Eq: Replacement rule}), we obtain
\begin{align}
&\frac{\de \kappa}{\de t}= \iota ( u_{1,0}' )
-\kappa \iota ( u_{1,1}' ) + ( \kappa -1 ) \iota ( u_{0,0}' ),\nonumber
\\
&\frac{\de \lambda}{\de t}=\iota ( u_{0,1}' )
-\lambda \iota ( u_{1,1}' ) + ( \lambda -1 ) \iota ( u_{0,0}' ). \label{Eq: Differential-difference example}
\end{align}
Replacing each $ \iota ( u_{i,j}' )$ in (\ref{Eq: Differential-difference example}) by its expression in terms of $\sigma$ gives
\begin{equation*}
\frac{\de \kappa}{\de t} = \mathcal{H}_{\kappa} \sigma, \qquad \frac{\de \lambda}{\de t} = \mathcal{H}_{\lambda} \sigma,
\end{equation*}
where
\begin{align}
&\mathcal{H}_{\kappa}= \frac{1 - \kappa}{\lambda_{1,0}} \s_{1} -
\frac{\kappa(\kappa-1)(\lambda_{1,0}-1)}{\lambda_{1,0}\kappa_{1,1}} \s_{1}\s_{2}+ (\kappa -1 )\id,\nonumber \\
&\mathcal{H}_{\lambda}= \frac{1 - \lambda}{\kappa_{0,1}} \s_{2} -
\frac{\lambda(\kappa-1)(\lambda_{1,0}-1)}{\lambda_{1,0}\kappa_{1,1}} \s_{1}\s_{2}+ (\lambda -1 )\id.\label{Eq: Linear difference operators running}
\end{align}
\end{Example}

\section{The invariant formulation of the Euler--Lagrange equations} \label{Section: The invariant formulation of the Euler--Lagrange equations}

Here we show how to calculate the Euler--Lagrange equations, in terms of invariants, for a given invariant difference Lagrangian.
Any such Lagrangian, $ \Lrm ({\bf n}, [{\bf u}] ) $, can be written in terms of the generating invariants $ \kappa^{\beta} $ and their shifts \smash{$ \kappa_{\bf J}^{\beta} = \s_{\bf J} \kappa^{\beta} $}:
\begin{equation}\label{LLk}
\Lrm ({\bf n}, [{\bf u}] ) = L^{\boldsymbol \kappa} ({\bf n}, [ {\boldsymbol \kappa} ] ).
\end{equation}
The key result is the following proposition, which generalizes O$\Delta$E invariantization \cite{mansfield2019movinga} to \pdese.

\begin{Proposition}[invariant Euler--Lagrange equations] \label{Prop: Invariant Euler--Lagrange equations}
Suppose that the Lagrangian $\Lrm(\mbn,[\mbu])$ is invariant under an $R$-parameter Lie group of point transformations, so that \eqref{LLk} holds.
Given the differential syzygies,
$\de\kappa^{\beta}/\de t=\mathcal{H}_{\alpha}^{\beta}\sigma^{\alpha}$, the following identity holds:
\begin{equation} \label{Eq: Invariantized Euler--Lagrange equations}
\mathrm{E}_{u^{\alpha}} (\Lrm) \frac{\de u^\alpha}{\de t} = \bigl( \bigl(\mathcal{H}_{\alpha}^{\beta}\bigr)^{\dagger} \mathrm{E}_{\kappa^{\beta}}\bigl(L^{\boldsymbol \kappa} \bigr) \bigr) {\sigma^{\alpha}},
\end{equation}
where
\begin{equation*}
\mathrm{E}_{\kappa^{\beta}} = \s_{-\bf J} \frac{\p}{\p \kappa_{\bf J}^{\beta}}
\end{equation*}
is the difference Euler operator with respect to $ \kappa^{\beta} $.
The invariantization of the original system of Euler--Lagrange equations is
\begin{equation} \label{Eq: Invariantized Euler--Lagrange equations 2}
\iota ( \mathrm{E}_{u^{\alpha}}( \Lrm)) = \bigl(\mathcal{H}_{\alpha}^{\beta}\bigr)^{\dagger} \mathrm{E}_{\kappa^{\beta}} ( L^{\boldsymbol \kappa})=0,\qquad \alpha =1, \dots ,q.
\end{equation}
\end{Proposition}

\begin{proof}
In the original coordinates,
\begin{equation}\label{Lu}
\frac{\de \Lrm}{\de t}=\frac{\p \Lrm}{\p \uJalpha} \frac{\de \uJalpha}{\de t}=\mathrm{E}_{u^\alpha}(\Lrm) \frac{\de u^\alpha}{\de t}+\operatorname{Div} (A_{\bf u} );
\end{equation}
here summation by parts has produced the difference divergence
\begin{equation}\label{Au}
\operatorname{Div} (A_{\bf u} )=\sum_\mbJ(\s_\mbJ-\id)\left(\s_{-\mbJ} \left(\frac{\p \Lrm}{\p \uJalpha}\right) (u^{\alpha} )'\right)=: \D_{n^{i}} \bigl(A_{\alpha}^{i} ({\bf n}, [{\bf u}]\bigr) (u^{\alpha} )' ),
\end{equation}
where each $A_\alpha^i$ is a linear difference operator. In the invariant coordinates,
\begin{align}
	\frac{\de L^{\boldsymbol \kappa}}{\de t}
	&=\frac{\p L^{\boldsymbol \kappa}}{\p \kappa^\beta_\mbK} \frac{\de \kappa^\beta_\mbK}{\de t}=\mathrm{E}_{\kappa^\beta}(L^{\boldsymbol \kappa}) \frac{\de \kappa^\beta}{\de t}+\operatorname{Div} (A_{\boldsymbol \kappa} )=\mathrm{E}_{\kappa^\beta}(L^{\boldsymbol \kappa}) \mathcal{H}^\beta_\alpha\sigma^\alpha+\operatorname{Div} (A_{\boldsymbol \kappa} )\nonumber\\
	&=\bigl\{\bigl(\mathcal{H}^\beta_\alpha\bigr)^\dagger\bigl(\mathrm{E}_{\kappa^\beta}(L^{\boldsymbol \kappa})\bigr)\bigr\}\sigma^\alpha+\operatorname{Div}(A_{\mathcal{H}}+A_{\boldsymbol \kappa}).\label{Lk}
\end{align}
Here, there are two contributions to the difference divergence:
\begin{gather*}
\operatorname{Div}(A_{\boldsymbol \kappa})=\sum_\mbK(\s_\mbK-\id)\left(\s_{-\mbK} \left(\frac{\p \Lrm}{\p \kappa^\beta_\mbK}\right)\bigl(\kappa^{\beta}\bigr)'\right)=: \D_{n^{i}} \bigl(F_{\beta}^{i}({\bf n}, [{\boldsymbol \kappa}]) \bigl(\kappa^{\beta}\bigr)'\bigr),\\
\operatorname{Div}(A_{\mathcal{H}})=\mathrm{E}_{\kappa^\beta}(L^{\boldsymbol \kappa}) \mathcal{H}^\beta_\alpha\sigma^\alpha-\bigl\{(\mathcal{H}^\beta_\alpha)^\dagger\bigl(\mathrm{E}_{\kappa^\beta}(L^{\boldsymbol \kappa})\bigr)\bigr\}\sigma^\alpha=: \D_{n^{i}} \bigl(H_{\alpha}^{i}({\bf n}, [{\boldsymbol \kappa}]) \sigma^{\alpha}\bigr),
\end{gather*}
where $F_{\beta}^{i}$ and $H_{\alpha}^{i}$ are linear difference operators. The difference between \eqref{Lu} and \eqref{Lk}, summed over $\mbn$ to annihilate the divergence terms, is
\[
0=\sum_\mbn\left(\frac{\de \Lrm}{\de t}-\frac{\de L^{\boldsymbol \kappa}}{\de t}\right)=\sum_\mbn\bigl(\mathrm{E}_{u^\alpha}(\Lrm) (u^\alpha)'-\bigl(\mathcal{H}^\beta_\alpha\bigr)^\dagger\bigl(\mathrm{E}_{\kappa^\beta}(L^{\boldsymbol \kappa})\bigr)\sigma^\alpha\bigr).
\]
Note that the invariants at a particular $\mbn$ are invariantized by the frame at that $\mbn$. Each~$u^\alpha(\mbn,t)$ has arbitrary (smooth) dependence on $t$; there is no link between $u^\alpha(\mbn,t)$ and $u^\alpha(\mathbf{m},t)$ for~${\mbn\neq\mathbf{m}}$. Therefore, \eqref{Eq: Invariantized Euler--Lagrange equations} holds for each $\mbn$, and so
\begin{equation*}
	\iota(\mathrm{E}_{u^\alpha}(\Lrm))\sigma^\alpha=\bigl(\mathcal{H}^\beta_\alpha\bigr)^\dagger\bigl(\mathrm{E}_{\kappa^\beta}(L^{\boldsymbol \kappa})\bigr)\sigma^\alpha.
\end{equation*}
Equation \eqref{Eq: Invariantized Euler--Lagrange equations 2} follows from the independence of the differential invariants $\sigma^\alpha$.
\end{proof}

\begin{Corollary}\label{CLrel}
	Under the conditions of Proposition {\rm \ref{Prop: Invariant Euler--Lagrange equations}}, using the notation in its proof,
	\begin{equation*}
		\operatorname{Div} (A_{\bf u} )=\operatorname{Div} (A_{\mathcal{H}}+A_{\boldsymbol \kappa} ).
	\end{equation*}
\end{Corollary}

\begin{proof}
	Compare \eqref{Lu} and \eqref{Lk}, taking \eqref{Eq: Invariantized Euler--Lagrange equations} into account.
\end{proof}
\begin{Example}[{Example \ref{Example: Normalization} cont.}]\label{Example: EL}
The Lagrangian (\ref{Eq: Lagrangian}) is written in terms of the generating invariants $ \kappa = \iota (u_{1,0})$ and $ \lambda = \iota(u_{0,1}) $ as
\begin{equation*}
L^{\boldsymbol \kappa} = \ln \big\vert \kappa -\lambda\big\vert.
\end{equation*}
Applying the Euler operators with respect to $\kappa$ and $\lambda$ gives
\begin{equation*}
\E_{\kappa} ( L^{\boldsymbol \kappa} ) =\frac{1}{\kappa - \lambda} ,\qquad
\E_{\lambda} ( L^{\boldsymbol \kappa} ) =\frac{-1}{\kappa - \lambda} .
\end{equation*}
From (\ref{Eq: Linear difference operators running}),
\begin{align*}
\mathcal{H}_{\kappa}^{\dagger} &= \frac{1 - \kappa_{-1,0}}{\lambda} \s_{1}^{-1} - \frac{\kappa_{-1,-1}(\kappa_{-1,-1}-1)(\lambda_{0,-1}-1)}{\kappa \lambda_{0,-1}} \s_{1}^{-1}\s_{2}^{-1}+ (\kappa -1 )\id, \\
\mathcal{H}_{\lambda}^{\dagger} &= \frac{1 - \lambda_{0,-1}}{\kappa} \s_{2}^{-1} - \frac{\lambda_{-1,-1}(\kappa_{-1,-1}-1)(\lambda_{0,-1}-1)}{\kappa \lambda_{0,-1}} \s_{1}^{-1}\s_{2}^{-1}+ (\lambda -1 )\id.
\end{align*}
By Proposition \ref{Prop: Invariant Euler--Lagrange equations}, the invariant Euler--Lagrange equation for the running example is
\begin{align*}
0&=\mathcal{H}_{\kappa}^{\dagger} \E_{\kappa} ( L^{\boldsymbol \kappa} )+\mathcal{H}_{\lambda}^{\dagger} \E_{\lambda} ( L^{\boldsymbol \kappa} )\\
 &= \frac{1 - \kappa_{-1,0}}{\lambda(\kappa_{-1,0}-\lambda_{-1,0})} -\frac{1 - \lambda_{0,-1}}{\kappa(\kappa_{0,-1}-\lambda_{0,-1})} - \frac{(\kappa_{-1,-1}-1)(\lambda_{0,-1}-1)}{\kappa \lambda_{0,-1}} +1.
\end{align*}
This is the invariantization of the Toda-type equation \eqref{TodaEL} given by the chosen normalization, which has the advantage that $L^{\boldsymbol{\kappa}}$ depends only on unshifted invariants.
\end{Example}

\section{Conservation laws} \label{Section: Conservation laws}

For ODEs \cite{gonccalves2012moving,gonccalves2013moving,mansfield2010practical}, PDEs \cite{gonccalves2016moving} and O$\Delta$Es \cite{mansfield2019movinga,mansfield2019movingb}, it has been shown that the conservation laws associated with an $R$-parameter Lie group of variational point symmetries are equivariant and can be written in terms of the invariants and the moving frame.
For P$\Delta$Es, we use the same reasoning as in the papers above to show that the $R$ conservation laws can be written in the equivariant form
\begin{equation*}
\D_{n^{i}} \big\lbrace V_{l}^{i}(\mbn,[\boldsymbol{\kappa}]) a_{r}^{l}(\rho)\big\rbrace=0, \qquad r=1,\dots, R,
\end{equation*}
where $a_{r}^{l}(\rho)$ are the components of the adjoint representation of $\rho $ and each $V_{l}^{i}(\mbn,[\boldsymbol{\kappa}])$ is invariant.

For a Lagrangian $\Lrm(\mbn,[\mbu])$ that is invariant under the one-parameter group generated by $\mbv_r$, the corresponding conservation law \eqref{Eq: Difference Noether's} given by Noether's theorem is
\begin{equation}\label{NoeCL}
\sum_{\bf J} ( \s_{\bf J}-\id )\left(Q_r^{\alpha}(\mbn,[\mbu]) {\s}_{-{\bf J}} \frac{\p \Lrm}{\p u_{{\bf J}}^{\alpha}} \right) =0.
\end{equation}
The left-hand side of \eqref{NoeCL} is almost the same as $\operatorname{Div}A_\mbu$ in \eqref{Au}, with the exception that $(u^\alpha)'$ in $A_\mbu$ is replaced by the characteristic component $Q_r^{\alpha}(\mbn,[\mbu])$. This replacement can be achieved by substituting the group parameter $\varepsilon^r$ for $t$ and evaluating the result at $\varepsilon^r=0$.

By Corollary \ref{CLrel}, the conservation law $\operatorname{Div}(A_\mbu)=0$ amounts to
\[
\operatorname{Div}(A_{\boldsymbol \kappa}+A_{\mathcal{H}})=\D_{n^{i}} \bigl(F_{\beta}^{i}({\bf n},[{\boldsymbol \kappa}]) \bigl(\kappa^{\beta}\bigr)'+H_{\alpha}^{i}({\bf n},[{\boldsymbol \kappa}])\sigma^{\alpha}\bigr)=0.
\]
If $t$ is the group parameter $ \varepsilon^r $ then $\bigl(\kappa^{\beta}\bigr)'= 0$, because each $\kappa^{\beta}$ is invariant.
Consequently, the conservation law given by Noether's theorem is
\begin{equation}\label{HCL}
	\D_{n^{i}} \bigl(H_{\alpha}^{i} ({\bf n}, [{\boldsymbol \kappa} ] )\sigma^{\alpha}\bigr)=0,
\end{equation}
where $\sigma^{\alpha}$ is the invariantization of the tangent vector to the group generated by $\mbv_r$, evaluated at $\varepsilon=0$.

\begin{Proposition}
Suppose that the conditions of Proposition {\rm \ref{Prop: Invariant Euler--Lagrange equations}} hold. If the linear difference operators in \eqref{HCL} are
\begin{equation*}
H_{\alpha}^{i} ({\bf n},{ [\boldsymbol \kappa ]} ) = \mathcal{C}^{i,\bf J}_{\alpha}({\bf n},{ [\boldsymbol \kappa ]}) \s_{\bf J},
\end{equation*}
then Noether's theorem gives the $R$ conservation laws,
\begin{equation*} 
\D_{n^{i}} \bigl\{\mathcal{C}_{\alpha}^{i,{\bf J}}({\bf n},{[\boldsymbol \kappa]}) \s_{\bf J}\lbrace \iota(Q^\alpha_s(\mbn, \mbu))a^s_r(\rho)\rbrace\bigr\} = 0,\qquad r=1,\dots, R.
\end{equation*}
\end{Proposition}

\begin{proof}
This proof is a slimmed-down analogue of its counterpart for O$\Delta$Es (see \cite{mansfield2019movinga}). We replace $t$ by the parameter $\varepsilon^r$ and let $\widehat{\mbu}(\mbn,\mbu,\varepsilon^r)$ be the orbit of the one-parameter local Lie group generated by $\mbv_r$. Let $\widehat{u}^\alpha_r$ denote $\de \widehat{u}^\alpha/\de \varepsilon^r$ and note that $\widehat{\mbu}(\mbn,\mbu,0)=\mbu$.
By the chain rule,
\[
\sigma^\alpha_r:=\iota (\widehat{u}^\alpha_r )=\left\{\frac{\p (g\cdot\widehat{u}^\alpha )}{\p \widehat{u}^\beta} \widehat{u}^\beta_r\right\}\bigg\vert_{g=\rho}.
\]
In particular, from \eqref{Eq: Adjoint identity 2},
\begin{align*}
\sigma^\alpha_r\big\vert_{\varepsilon^r=0}&=\left\{\frac{\p (g\cdot u^\alpha )}{\p u^\beta} Q^\beta_r(\mbn,\mbu)\right\}\bigg\vert_{g=\rho}
	=\bigl\{Q^\alpha_s(\mbn,g\cdot \mbu) a^s_r(g)\bigr\}\big\vert_{g=\rho}=\iota (Q^\alpha_s(\mbn, \mbu) )a^s_r(\rho).
\end{align*}
Substitute $\sigma^\alpha_r\big\vert_{\varepsilon^r=0}$ for $\sigma^\alpha$ in \eqref{HCL} to complete the proof.
\end{proof}

By the prolongation formula, the conservation laws amount to
\begin{equation*}
\D_{n^{i}} \bigl\{\mathcal{C}_{\alpha}^{i,{\bf J}}(\mbn,[\boldsymbol{\kappa}]) (\s_{\bf J} \iota (Q_{s}^{\alpha}))a_{r}^{s}(\rho_{\bf J})\bigr\}=0.
\end{equation*}
The adjoint representation is a Lie group representation, so
\begin{equation*}
a^{s}_{r}(\rho_{\mbJ})
= a^{s}_{l} \bigl(\rho_{\bf J}\rho^{-1}\bigr)
a_{r}^{l}(\rho)= a^{s}_{l} (\iota(\rho_{\bf J}))
a_{r}^{l}(\rho).
\end{equation*}
This leads to the following corollary.
\begin{Corollary}
The conservation laws for a difference frame may be written in the form
\begin{equation*}
\D_{n^{i}} \big\lbrace V_{l}^{i}(\mbn,[\boldsymbol{\kappa}]) a_{r}^{l}(\rho) \big\rbrace =0,
\end{equation*}
where
\begin{equation*}
V_{l}^{i}(\mbn,[\boldsymbol{\kappa}]) = \mathcal{C}_{\alpha}^{i,{\bf J}}(\mbn,[\boldsymbol{\kappa}]) (\s_{\bf J} \iota (Q_{s}^{\alpha}) ) a^{s}_{l} (\iota(\rho_{\bf J})).
\end{equation*}
\end{Corollary}

\begin{Example}[{Example \ref{Example: Normalization} cont.}]
For the running example, the equivariant conservation laws are obtained as follows:
\begin{align*}
\operatorname{Div}(A_{\mathcal{H}})
={}&\E_{\kappa}(L^{\boldsymbol \kappa})\mathcal{H}_{\kappa} \sigma
+E_{\lambda}(L^{\boldsymbol \kappa})\mathcal{H}_{\lambda} \sigma
-\bigl\{(\mathcal{H}_{\kappa})^{\dagger}\E_{\kappa}(L^{\boldsymbol \kappa})
+(\mathcal{H}_{\lambda})^{\dagger}\E_{\lambda}(L^{\boldsymbol \kappa})\bigr\}\sigma\\
={}& (\s_{1} - \id) \left(\frac{1-\kappa_{-1,0}}{\lambda (\kappa_{-1,0}-\lambda_{-1,0})} \sigma \right) + (\s_{2}-\id) \left(\frac{\lambda_{0,-1}-1}{\kappa (\kappa_{0,-1}-\lambda_{0,-1})} \sigma \right)\\
& + (\s_{1}\s_{2} - \id) \left(- \frac{ (\kappa_{-1,-1}-1 ) (\lambda_{0,-1}-1 )}{\kappa \lambda_{0,-1}} \sigma\right)\\
={}&\D_{n^{1}}\left\{\frac{1-\kappa_{-1,0}}{\lambda (\kappa_{-1,0}-\lambda_{-1,0})} \sigma- \frac{(\kappa_{-1,0}-1)(\lambda-1)}{\kappa_{0,1} \lambda} \s_2\sigma \right\}\\
& +\D_{n^{2}}\left\{\frac{\lambda_{0,-1}-1}{\kappa (\kappa_{0,-1}-\lambda_{0,-1})} \sigma- \frac{(\kappa_{-1,-1}-1)(\lambda_{0,-1}-1)}{\kappa \lambda_{0,-1}} \sigma
\right\}.
\end{align*}
The invariantized (unprolonged) infinitesimals are $\iota (Q_{1} )=1 $ and $\iota (Q_{2} )=0$, so
\begin{equation*}
\iota (Q_s(\mbn, \mbu) )a^s_r(\rho)=a^1_r(\rho).
\end{equation*}
The adjoint action of $g\colon u\mapsto \widetilde{u}=bu+a$ on the infinitesimal generators gives $\mbv_1=b\widetilde{\mbv}_1$ and~${\mbv_2=-a\widetilde{\mbv}_1+\widetilde{\mbv}_2}$,
so the components of the adjoint matrix are
\[
a_1^1=b,\qquad a_1^2=0,\qquad a_2^1=-a,\qquad a_2^2=1.
\]
On the frame $\rho$ (at $(m,n)$), these components are
\begin{equation}\label{adframe}
a_1^1(\rho)=\frac{1}{u_{1,1}-u_{0,0}} ,\qquad a_1^2(\rho)=0,\qquad a_2^1(\rho)=\frac{u_{0,0}}{u_{1,1}-u_{0,0}} ,\qquad a_2^2(\rho)=1.
\end{equation}
Applying $\s_2$ to \eqref{adframe} gives the components on the frame at $(m,n+1)$. In particular,
\begin{align*}
a_1^1(\s_2\rho)
&=\frac{1}{u_{1,2}-u_{0,1}}= \frac{\kappa_{0,1}}{1-\lambda} a_1^1(\rho),\\
a_2^1(\s_2\rho)
&=\frac{u_{0,1}}{u_{1,2}-u_{0,1}}= \frac{\kappa_{0,1}}{1-\lambda} a_2^1(\rho)+\frac{\lambda\kappa_{0,1}}{1-\lambda} a_2^2(\rho).
\end{align*}
Consequently, in terms of the frame \eqref{adframe}, the equivariant versions of the conservation laws arising from Noether's theorem are
\begin{gather*}
	0=\D_{n^{1}}\left\{\left(\frac{1-\kappa_{-1,0}}{\lambda (\kappa_{-1,0}-\lambda_{-1,0})} + \frac{\kappa_{-1,0}-1}{ \lambda}\right)a_1^1(\rho) \right\}\\
\hphantom{0=}{} +\D_{n^{2}}\left\{\left(\frac{\lambda_{0,-1}-1}{\kappa (\kappa_{0,-1}-\lambda_{0,-1})} - \frac{(\kappa_{-1,-1}-1)(\lambda_{0,-1}-1)}{\kappa \lambda_{0,-1}}\right) a_1^1(\rho)
	\right\},\\
	0= \D_{n^{1}}\left\{\left(\frac{1-\kappa_{-1,0}}{\lambda (\kappa_{-1,0}-\lambda_{-1,0})} + \frac{\kappa_{-1,0}-1}{\lambda}\right)a_2^1(\rho)+(\kappa_{-1,0}-1) a_2^2(\rho)
	 \right\}\\
\hphantom{0=}{}+\D_{n^{2}}\left\{\left(\frac{\lambda_{0,-1}-1}{\kappa (\kappa_{0,-1}-\lambda_{0,-1})} - \frac{(\kappa_{-1,-1}-1)(\lambda_{0,-1}-1)}{\kappa \lambda_{0,-1}}\right) a_2^1(\rho)
\right\}.
\end{gather*}
\end{Example}

\section{Differential-difference structure} \label{Section: Differential-difference structure}

Differential-difference moving frames are set in a continuous space that embodies prolongation with respect to both derivatives and shifts (see \cite{peng2022transformations}). Just as for \pdese, the discrete independent variables, $\mbn =\bigl(n^1,\dots,n^m\bigr)$, and the corresponding shift operators are invariant. In general, the invariant derivative operators do not commute with one another or with the shift operators, a~problem that can be resolved with substantial technical machinery.\footnote{Moving frames for differential Euler--Lagrange equations can be expressed naturally in terms of the variational bicomplex \cite{kogan2003invariant}; for a constructive approach that lends itself to symbolic computation, see~\cite{gonccalves2016moving}.} However, such machinery is not needed for systems of \ddes for $\mbu=\bigl(u^1,\dots,u^q\bigr)\in\bbr^q$ with just one continuous independent variable, $x$, provided that the group action on $x$ is sufficiently simple. For consistency with our presentation of difference moving frames, we restrict attention to such systems and actions.

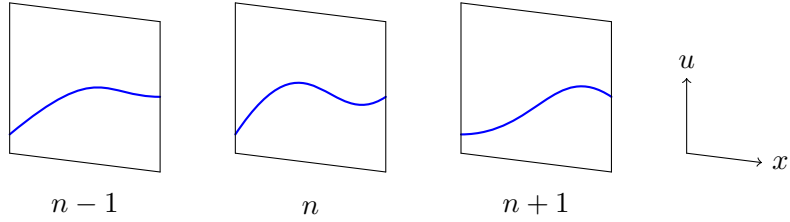
\begin{figure}[t]\centering
		\begin{tikzpicture}[line cap=round,line join=round,>=triangle 45]
			\draw[blue, thick] (0,0.5) .. controls (1,2) and (1.25,0.5).. (2,1);
			\draw (1,-0.7) node[anchor=south] {$n$};
			\draw (0,0.25) -- (2,0);
			\draw (2,0) -- (2,2);
			\draw (2,2) -- (0,2.25);
			\draw (0,2.25) -- (0,0.25);
			\draw[blue, thick] (3,0.5) .. controls (4,0.5) and (4.25,1.5).. (5,1);
			\draw (4,-0.7) node[anchor=south] {$n+1$};
			\draw (3,0.25) -- (5,0);
			\draw (5,0) -- (5,2);
			\draw (5,2) -- (3,2.25);
			\draw (3,2.25) -- (3,0.25);
			\draw[blue, thick] (-3,0.5) .. controls (-1.8,1.5) and (-1.75,1).. (-1,1);
			\draw (-2,-0.7) node[anchor=south] {$n-1$};
			\draw (-3,0.25) -- (-1,0);
			\draw (-1,0) -- (-1,2);
			\draw (-1,2) -- (-3,2.25);
			\draw (-3,2.25) -- (-3,0.25);
			\draw (7,0.125) node[anchor=west] {$x$};
			\draw (6,1.25) node[anchor=south] {$u$};
			\draw[-to] (6,0.25) -- (7,0.125);
			\draw[-to] (6,0.25) -- (6,1.25);
		\end{tikzpicture}
	\caption{A graph on (ordered) slices in the total space $\mathcal{T} = \bbr \times \bbz \times \bbr$.}
	\label{Fig: Differential-difference structure}
\end{figure}

The total space,
$\mathcal{T}=\bbr \times \bbz^m \times \bbr^q$, consists of a continuous \textit{slice}, ${\mathcal{T}}_\mbn = \bbr \times \lbrace{\mbn}\rbrace \times\bbr^q$, over each~$\mbn\in\bbz^m$. Consequently, each graph on $\mathcal{T}$ defined by $\mbu=\mathbf{f}(x,\mbn)$ restricts to a graph on every slice, as illustrated in Figure \ref{Fig: Differential-difference structure}. We consider only graphs that are smooth on each slice, which can be prolonged by differentiation as many times as needed.\footnote{If graphs are only locally smooth, restrict attention to neighbourhoods in which they are smooth.} The differential prolongation structure on the slice over any $\mbn$ is embodied by the infinite jet space, $J^\infty(\mathcal{T}_\mbn)$. This space has vertical coordinates $u^\alpha_{j;\mbzero}$, where $j$ denotes the number of derivatives with respect to $x$. In particular, $u^\alpha_{0;\mbzero}$ represents $u^\alpha$ at $\mbn$.

The differential prolongation structure over any fixed $\mbn$ is replicated over all other points ${\mbn+\mbK}$.
Consequently, one can use difference prolongation to construct the (differential-dif\-fer\-ence) \textit{prolongation space} over $\mbn$, denoted $P(J^\infty(\mcT_\mbn))$. This space has coordinates $(x,(u^\alpha_{j;\mbK}))$, where $u^\alpha_{j;\mbK}$ represents the value of $u^\alpha_{j;\mbzero}$ on the jet space over $\mbn+\mbK$.

Just as for difference equations, the horizontal translation $\mathrm{T}_{\mbI}$ maps $\mbn$ to $\mbn+\mbI$ without changing any other coordinates. The pullback of
\[
\mathrm{T}_\mbI\colon\ \prol\rightarrow P(J^\infty(\mcT_{\mbn+\mbI}))
\]
to $\prol$ is the shift operator $\s_{\mbI}$, which can be regarded as a product of unit forward shifts~$\s_i$, $i=1,\dots,m$, and their inverses. This acts on functions $f\in C^\infty(\prol)$ as follows:
\begin{equation*}
	\s_{\mbI} f( x,\mbn , ( u_{j;\mbK}^{\alpha} ) ) = f (x, \mbn + \mbI , ( u_{j;\mbK+ \mbI}^{\alpha} ) ).
\end{equation*}
All shift operators $\s_\mbK$ commute with one another and with the total derivative operator,
\begin{equation*}
D = \frac{\p}{\p x} + u_{j+1;k} \frac{\p}{\p u_{j;k}} .
\end{equation*}
For notational consistency, we define $D_{(j)}:=D^j$,
so the vertical coordinates on $\prol$ may be written as $u_{j;\mbK}=\s_{\mbK} D_{(j)} u^\alpha_{0;0}$.

The differential-difference divergence of a given $(1+m)$-tuple of functions in $C^\infty(\prol)$, $A=\bigl(A^0;A^1,\dots,A^m\bigr)$, is
\[
\operatorname{Div}(A):=DA^0 +\D_{n^i}A^i,
\]
where $\D_{n^i}= \s_i-\id$ is the $i^{\mathrm{th}}$ forward difference operator. Note: we sum over $i$ from $1$ to $m$ only. A linear differential-difference operator on the prolongation space $\prol$ is an operator of the form $\mcH=h^{j;\mbK}\s_{\mbK} D_{(j)}$, for given functions $h^{j;\mbK}$. The formal adjoint of $\mathcal{H}$ is the unique operator $\mathcal{H}^{\dagger}$ such that
\begin{equation*}
	f\mathcal{H}g - \bigl(\mathcal{H}^{\dagger} f\bigr) g\in \operatorname{im}(\operatorname{Div}).
\end{equation*}
Using the standard identities
\[
\s_{\mbK}^{\dagger}=\s_{-\mbK}, \qquad D_{(j)}^{\dagger} = (-D)_{(j)}:=(-1)^{j}D^j,
\]
one obtains
\[
\mathcal{H}^{\dagger} f=(-D)_{(j)}\s_{-\mbK} \bigl(h^{j;\mbK} f\bigr).
\]

\begin{Definition}
	A \textit{conservation law} of a given system of \ddes is a differential-difference divergence expression, $\mathcal{C}=\operatorname{Div} (A )$, such that $\mathcal{C}=0$ on all solutions of the system.
\end{Definition}

A given \dde defines a variety in the continuous space $\prol$, on which it is possible to construct moving frames that respect both the differential and difference structures. Every Lie group, $G$, of point transformations of the total space, whose prolongation to $\prol$ preserves these structures, consists of \textit{projectable} diffeomorphisms on each slice:
\[
g\colon\ \mathcal{T} \rightarrow \mathcal{T}, \qquad g \colon\ (x, \mbn, \mbu ) \mapsto (g\cdot x, g\cdot \mbn, g\cdot\mbu ):= (\widetilde{x}(x), \mbn, \widetilde{\mbu}(x,\mbn,\mbu) ).
\]
The projectability condition arises because mappings for which $\widetilde{x}$ depends on $\mbn$ or $\mbu$ are incompatible with the prolongation structure (see \cite{peng2022transformations} for details). This is a key distinction between~\ddes and PDEs, for which Lie point transformations do not have to be projectable.
For each $\mbn$, the mapping $g$ is a diffeomorphism, and therefore the Jacobian determinants,
\begin{equation*}
J_x:=\frac{\de \widetilde{x}}{\de x} ,\qquad J_\mbu:=\det \left(\frac{\p \widetilde{u}^\alpha}{\p u^\beta}\right),
\end{equation*}
are nonzero. The prolongation conditions give the action of $g$ on $\prol$ (recursively):
\[
g\cdot u^\alpha_{j+1;\mbzero}=\widetilde{u}^\alpha_{j+1;\mbzero}:=\frac{D \widetilde{u}^\alpha_{j;\mbzero}}{D\widetilde{x}} ,\qquad g\cdot u^\alpha_{j;\mbK}=\widetilde{u}^\alpha_{j;\mbK}:=\s_{\mbK} \widetilde{u}^\alpha_{j;\mbzero}.
\]
A tilde $\widetilde{\phantom{v}}$ over a function or operator denotes that $x$ and each $u^\alpha_{j;\mbK}$ are replaced by $\widetilde{x}$ and $\widetilde{u}^\alpha_{j;\mbK}$ respectively.

From here on, we consider $R$-parameter Lie groups of (projectable) point transformations, whose generators are of the form
\[
\mbv_r=\xi_r(x)\p_{x}+\eta^\alpha_r(x,\mbn,\mbu)\p_{u^\alpha} ,\qquad r=1,\dots,R.
\]
The adjoint action of $g$ on the Lie algebra spanned by $\mbv_1,\dots,\mbv_R$ satisfies the identities
\begin{equation}\label{ddad}
	\mbv_r=a^s_r(g)\widetilde{\mbv}_s ,\qquad r=1,\dots,R.
\end{equation}
The characteristic corresponding to $\mbv_r$ is $\mathbf{Q}_r=\bigl(Q^1_r,\dots,Q^q_r\bigr)$, where
\[ Q^\alpha_r=\eta^\alpha_r(x,\mbn,\mbu)-\xi_r(x) u^\alpha_{1;\mbzero} .
\]
This enables the differential-difference prolongation of $\mbv_r$ to be written as
\begin{equation}\label{prolv}
\operatorname{pr} \mbv_r=\xi_rD+X_r,\qquad \text{where}\quad X_r:=(\s_\mbK D_{(j)} Q^\alpha_r)\p_{u^\alpha_{j;k}} .
\end{equation}
The operator $X_r$ is the characteristic form of the generator; it acts only on the vertical coordinates of $\prol$. From here on, we write $\operatorname{pr} \mbv_r$ simply as $\mbv_r$; generators are assumed to be prolonged wherever this is needed. Note that \eqref{ddad} holds equally for the prolonged generators.

\begin{Lemma}
	Let G be a Lie group of projectable transformations of the total space $\mcT$, and let~$g\in G$. Using the above notation, the following identities hold:
	\begin{gather}
		\xi_rD=a_{r}^{s}(g) \widetilde{\xi}_s\widetilde{D};\label{ddadxid}\\
		X_r=a_{r}^{s}(g)\widetilde{X}_s;\label{ddadX}\\
		J_x \xi_r=\widetilde{\xi}_s a_{r}^{s}(g);\label{ddadxi}\\
		\left(\frac{\p \widetilde{u}^{\alpha}}{\p u^{\beta}}\right) Q_{r}^{\beta}= \widetilde{Q}_{s}^{\alpha} a_{r}^{s}(g).\label{ddadQ}
	\end{gather}
\end{Lemma}
\begin{proof}
	The chain rule and \eqref{ddad} give
	\[
	\xi_rD=\xi_r \frac{\de \widetilde{x}}{\de x} \widetilde{D}=\mbv_r(\widetilde{x})\widetilde{D}=a^s_r(g)\widetilde{\mbv}_s(\widetilde{x})\widetilde{D}=a^s_r(g) \widetilde{\xi_s}\widetilde{D} ,
	\]
	which proves \eqref{ddadxid}. To obtain \eqref{ddadX}, substitute \eqref{prolv} into the prolongation of \eqref{ddad} and take \eqref{ddadxid} into account. Then apply \eqref{ddadxid} to $\widetilde{x}$ and \eqref{ddadX} to $\widetilde{u}^{\alpha}$ to prove \eqref{ddadxi} and \eqref{ddadQ}.
\end{proof}

The construction of a differential-difference moving frame is essentially the same as that of a difference moving frame.
Suppose that the Lie group $G$ acts smoothly, freely and regularly on $\mathcal{M}\subset\prol$, a finite prolongation space over $\mbn$ whose coordinates include all relevant variables, including at least one derivative \smash{$u^\alpha_{j;\mbK}$}, $j\geq 1$.
Let the cross-section and frame on~$\mathcal{M}$ be~$\mathcal{K}$ and $\rho$, respectively.
Such a moving frame uses the $\mbK^{\mathrm{th}}$ translate of $\mathcal{K}$ at every other base point $\mbn+\mbK$; the cross-section and frame at $\mbn+\mbK$ are represented on $\mathcal{M}$ by $\s_\mbK\mathcal{K}$ and~$\rho_{0;\mbK}=\s_\mbK\rho$, respectively. This construction extends immediately to $\prol$, because $\mathcal{M}$ has arbitrary dimension.

As usual, the invariantization (denoted by $\iota$) of a function, operator, etc., is obtained by evaluating the transformed quantity on the frame $\rho$.
The freedom to choose a cross-section leads to the possibility that $\iota(x)$ may depend on one or more of the variables $u^\alpha_{j;\mbK}$, because~$\widetilde{x}=g\cdot x$ depends on $x$ and the group parameters.

\begin{Lemma}
	The invariantized total derivative, $\mcD=\iota(D)$, commutes with all $($invariant$)$ shift operators $\s_\mbK$ if and only if $\mcJ=J_x|_{g=\rho} $ is a function of $x$ alone.
\end{Lemma}

\begin{proof}
	Evaluating the identity $D=J_x \widetilde{D}$ on $\rho$ gives $D=\mcJ\mcD$. Each $\s_i$ commutes with $D$, so
	\[
	[\mcD,\s_i]=\mcJ^{-1}D\s_i-\bigl(\s_i\mcJ^{-1}\bigr)\s_iD=\mcJ^{-1} [D,\s_i]-\bigl(D_{n^i}\mcJ^{-1}\bigr)\s_iD=-\bigl(D_{n^i}\mcJ^{-1}\bigr)\s_iD.
	\]
	The right-hand side is zero if and only if $\s_i\mcJ=\mcJ$, which requires that $\mcJ$ is independent of $n^i$ and $[\mbu]$. So a necessary and sufficient condition for every shift operator to commute with $\mcD$ is that $\mcJ$ depends on neither $\mbn$ nor $[\mbu]$.
\end{proof}

Every group parameter that occurs in $J_x$ also occurs in $\widetilde{x}$, which may also have one further parameter associated with translations in $x$. So a sufficient condition for $\mcD$ to commute with all shifts is that $\iota(x)$ depends on $x$ alone.\footnote{The only exception to this condition being necessary occurs when $G$ includes translations in $x$. As an example, consider $(\widetilde{x},\widetilde{u})=(x+a,e^au+b)$. The unusual normalization $\iota(u)=0,\ \iota(u_{1;\mathbf{0}})=1$ gives $a=-\ln (u_{1;\mathbf{0}})$, so $\iota(x)=x-\ln(u_{1;\mathbf{0}})$, but $\mcD$ commutes with all shifts because $\mcJ=1$. A normalization for which $\iota(x)$ depends on $x$ only is $\iota(x)=0$, $\iota(u)=0$.} Indeed, even if $\widetilde{x}$ includes translations, the requirement for the parameters in $J_x$ to depend on $x$ alone (on $\rho$) ensures that there exist normalizations for which~$\iota(x)$ depends on $x$ only: simply set $\iota(x)$ to be constant. These observations motivate the following general definition.

\begin{Definition} \label{Def: Projectable frame}
	Given a prolongation space\footnote{This definition includes jet spaces for differential equations, as projectability is also useful in this context.} $\mathcal{P}$ on which the Lie group $G$ acts smoothly, freely and regularly, let $\mathcal{B}$ denote the space coordinatized by the continuous independent variables, $\mbx=\bigl(x^1,\dots,x^p\bigr)$. A moving frame, $\rho$, is \textit{projectable} if $\iota(\mbx)=\rho\cdot\mbx$ is a function of $\mbx$ alone.
\end{Definition}

If $G$ consists of projectable transformations, then $\mathcal{B}$ is invariant. Thus, one can restrict attention to the action of $G$ on $\mathcal{B}$, ignoring those elements of $G$ that fix every $\mbx\in\mathcal{B}$. Let~${G_\mathcal{B}\colon\mathcal{B}\rightarrow \mathcal{B}}$ denote the resulting Lie group of restricted transformations.

\begin{Lemma}Let $G$ be a Lie group of projectable transformations that acts smoothly, freely and regularly on a prolongation space~$\mathcal{P}$. Then there exists a projectable moving frame on $\mathcal{P}$ if and only if $G_\mathcal{B}$ acts freely and regularly on $\mathcal{B}$.
\end{Lemma}

\begin{proof}Suppose that $G_\mathcal{B}$ acts freely and regularly on $\mathcal{B}$; smoothness is inherited from $G$. A~moving frame $\rho_{\mathcal{B}}$ on $\mathcal{B}$ can be constructed by imposing normalization conditions. This determines the group parameters that occur in $G_{\mathcal{B}}$ and ensures that $\rho_{\mathcal{B}}\cdot\mbx$ depends on $\mbx$ alone. Then $\rho_{\mathcal{B}}$ can be extended to a moving frame $\rho$ on $\mathcal{P}$ by choosing a normalization that determines the remaining parameters in $G$. Conversely, if there exists a projectable moving frame $\rho$ on $\mathcal{P}$, its restriction to $\mathcal{B}$ is a moving frame, whose existence requires $G_\mathcal{B}$ to act freely and regularly.
\end{proof}

In particular, if $x$ is the only continuous independent variable, a projectable moving frame exists only if $G_{\mathcal{B}}$ depends on at most one parameter, for otherwise the restricted action is not free. The moving frame is projectable if either $x$ is invariant, or $\widetilde{x}$ depends on just one group parameter that is determined by a normalization equation of the form $\iota(x)=\mathrm{const}$. Either of these conditions ensures that $\mathcal{D}$ commutes with $\s_\mbK$. We restrict attention to projectable moving frames for the remainder of this paper.

\section{Differential-difference variational calculus}
\label{Section: The differential-difference calculus of variations}

The \dde variational calculus in terms of invariants is derived in much the same way as its counterpart for \pdese, so we present the basic method concisely to avoid too much repetition. However, there are some important differences that stem from the group action on the continuous independent variable~$x$; we describe these in detail. Here and henceforth, square brackets around an expression denotes the expression and finitely many of its derivatives and shifts.

Let $G$ be an $R$-parameter Lie group of variational point symmetries for a given Lagrangian functional,
\[
\mathcal{L}=\sum_\mbn \int\Lrm (x,\mbn, [\mbu ] ) \de x,
\]
that leave the one-form $\Lrm (x,\mbn, [\mbu ] ) \de x$ invariant under the group action. Then
\begin{equation}\label{Ldxinv}
	\Lrm \de x=\widetilde{\Lrm} \de \widetilde{x}= \widetilde{\Lrm} J_x \de x.
\end{equation}
In particular, invariance under the transformations generated by $\mbv_r$ amounts to the condition
\begin{equation*}
\mbv_r(\Lrm)+\Lrm D\xi_r=0
\end{equation*}
(see \cite{olver2000applications} for details). A useful equivalent form of this condition is
\begin{equation}\label{ddQinv}
	X_r(\Lrm)+D(\Lrm \xi_r)=0,
\end{equation}
where $X_r$ is the generator in characteristic form (whose action on $x$ and $\de x$ is trivial).

The Euler--Lagrange equations are
\[
\E_{u^{\alpha}}(\Lrm):= \s_{-\mbK}(-D)_{(j)}\left(\frac{\p\Lrm}{\p u^\alpha_{j;\mbK}}\right)=0,\qquad \alpha = 1,\dots,m.
\]
Just as for \pdese, these equations can be obtained by allowing $[\mbu]$ (but \textit{not} $x$) to depend smoothly on a continuous invariant parameter $t\in\bbr$. Then, using $\phantom{.}'$ to denote the derivative with respect to $t$,
\begin{align}
\frac{\de}{\de t} ( \Lrm \de x )
&= \frac{\p \Lrm}{\p u^\alpha_{j;k}}
 (u^\alpha_{j;k})' \de x= \E_{u^\alpha} (\Lrm ) (u^\alpha_{0;\mbzero})' \de x
+ \operatorname{Div} (A_{\mbu} ) \de x,\label{ddEL2}
\end{align}
where $\operatorname{Div} (A_{\mbu} )$ consists of the divergence terms arising from the summation and integration by parts. Note that $X_r(\Lrm)$ is obtained from $\de\Lrm/\de t$ by setting $ [(u^\alpha_{0;\mbzero})'=Q^\alpha_r ]$. So, from \eqref{ddQinv} and~\eqref{ddEL2}, the Noether conservation law corresponding to invariance under $\mbv_r$ is
\begin{align} 0=-Q^\alpha_r\E_{u^\alpha}(\Lrm)&=\operatorname{Div}(A_{\mbu})\big|_{[(u^\alpha_{0;\mbzero})'=Q^\alpha_r]}-X_r(\Lrm)=\operatorname{Div}(A_{\mbu})\big|_{[(u^\alpha_{0;\mbzero})'=Q^\alpha_r]}+D(\Lrm \xi_r).\label{dduNoe}
\end{align}

Evaluating \eqref{Ldxinv} on the frame gives
\begin{equation*}
\Lrm \de x= \iota(\Lrm) \iota(\de x)=:L^{\boldsymbol{\kappa}}(\iota(x),\mbn,[\boldsymbol{\kappa}]) \iota(\de x),
\end{equation*}
where $\boldsymbol{\kappa}$ denotes the generating invariants that depend on $[\mbu]$ (and possibly also on $x$). Explicitly, $\iota(\de x)=\mcJ \de x$, so $L^{\boldsymbol{\kappa}}=\Lrm\mcJ^{-1}$. As the frame is projectable, $\mcJ$ depends only on $x$; in particular, if $x$ is invariant, $\mcJ=1$.

To obtain the invariantized Euler--Lagrange equations, differentiate the one-form $L^{\boldsymbol{\kappa}}\iota(\de x)$ with respect to $t$. Each $\kappa^\beta$ depends on $t$ through its dependence on $[\mbu]$, giving rise to the syzygies
$\de\kappa^{\beta}/\de t=\mathcal{H}_{\alpha}^{\beta}\sigma^{\alpha}$, where $\sigma^\alpha=\iota((\uoo^\alpha)')$ and each $\mcH_\alpha^\beta$ is a an invariant differential-difference operator. Specifically,
\[
\mathcal{H}_{\alpha}^{\beta}=\iota\left(\frac{\p \kappa^\beta}{\p u^\alpha_{j;\mbK}}\right) \mcD_{(j)}\s_\mbK,
\]
where $\mcD=\mcJ^{-1}D$ is the invariantized total derivative and \smash{$\mcD_{(j)}:=\mcD^j$}. Note that $\mcD$ commutes with each $\s_\mbK$; thus the invariant derivatives and shifts of $\kappa^\beta$ can be written as $\kappa^\beta_{j;\mbK}:=\mcD_{(j)}\s_\mbK\kappa^\beta$.
Integration by parts is straightforward: given two functions, $f$ and $g$,
\[
f(\mcD g) \iota(\de x)=f(Dg) \de x=\{D(fg)-(Df)g\} \de x =\{\mcD(fg)-(\mcD f)g\}\iota(\de x).
\]
As $\mcD(fg) \iota(\de x)$ is a divergence, define the adjoint of $\mcD$ relative to $\iota(\de x)$ to be $\mcD^{\boldsymbol{\dagger}}=-\mcD$. (The bold dagger $\boldsymbol{\dagger}$ distinguishes this adjoint from the standard adjoint, $\dagger$.) As the frame is projectable, $\s_\mbK^{\boldsymbol{\dagger}}=\s_\mbK^{\dagger}=\s_{-\mbK}$.
Therefore, using the notation $(-\mcD)_{(j)}:=(-1)^j\mcD_{(j)}$, we obtain
\begin{align}
	\frac{\de }{\de t} (L^{\boldsymbol \kappa}\iota(\de x) t)
	&=\frac{\p L^{\boldsymbol \kappa}}{\p \kappa^\beta_{j;\mbK}} \frac{\de \kappa^\beta_{j;\mbK}}{\de t} \iota(\de x)\nonumber\\
	&=\left((-\mcD)_{(j)}\s_{-\mbK} \frac{\p L^{\boldsymbol \kappa}}{\p \kappa^\beta_{j;\mbK}}\right)\frac{\de \kappa^\beta}{\de t} \iota(\de x)+\mcD\mathrm{iv} (A_{\boldsymbol \kappa} )\iota(\de x)\nonumber\\
	&=\mathrm{E}_{\kappa^\beta}(L^{\boldsymbol \kappa})\bigl(\mathcal{H}^\beta_\alpha\sigma^\alpha\bigr)\iota(\de x)+\mcD\mathrm{iv}(A_{\boldsymbol \kappa})\iota(\de x)\nonumber\\
	&=\bigl\{\bigl(\mathcal{H}^\beta_\alpha\bigr)^{\boldsymbol{\dagger}}(\mathrm{E}_{\kappa^\beta}(L^{\boldsymbol \kappa}))\bigr\}\sigma^\alpha\iota(\de x)+\mcD\mathrm{iv}(A_{\mathcal{H}}+A_{\boldsymbol \kappa})\iota(\de x).\label{ddLk}
\end{align}
The one-forms $\mcD\mathrm{iv} (A_{\boldsymbol \kappa} )\iota(\de x)$ and $\mcD\mathrm{iv}(A_\mcH)\iota(\de x)$ are defined by the above; $\mcD\mathrm{iv}$ has the same form as $\operatorname{Div}$, but
with $D$ replaced by $\mcD$. The following identity is useful:
\begin{equation}\label{Divid}
\mcD\mathrm{iv}\bigl(A^0;A^1,\dots,A^m\bigr) \iota(\de x)=\operatorname{Div}\bigl(A^0;\mcJ A^1,\dots,\mcJ A^m\bigr) \de x.
\end{equation}
Similarly to \pdese, the $i^{\mathrm{th}}$ component of $A_{\boldsymbol \kappa}$ (resp.\ $A_\mcH$) is of the form $F^i_\beta \bigl(\kappa^\beta\bigr)'$ \big(resp.\ $H^i_\alpha\sigma^\alpha$\big); here $F^i_\beta$ and $H^i_\alpha$ are invariant differential-difference operators.

\begin{Proposition}\label{ddEL}
	Suppose that the Lagrangian one-form $\Lrm(\mbn,x,[\mbu]) \de x$ is invariant under an $R$-parameter Lie group of point transformations.
	In the above notation,
	\begin{equation} \label{ddELid}
		\mathrm{E}_{u^{\alpha}} ( \Lrm ) (u^\alpha_{0;\mbzero})' \de x = \bigl( \bigl(\mathcal{H}_{\alpha}^{\beta}\bigr)^{\boldsymbol{\dagger}} \mathrm{E}_{\kappa^{\beta}} (L^{\boldsymbol \kappa} ) \bigr) {\sigma^{\alpha}}\iota(\de x),
	\end{equation}
	so the invariantized Euler--Lagrange \ddes are
	\begin{equation} \label{ddinvEL}
		\iota ( \mathrm{E}_{u^{\alpha}}( \Lrm )) = \bigl(\mathcal{H}_{\alpha}^{\beta}\bigr)^{\boldsymbol{\dagger}} \mathrm{E}_{\kappa^{\beta}} ( L^{\boldsymbol \kappa})=0,\qquad \alpha =1, \dots ,q.
	\end{equation}
	Furthermore,
	\begin{equation}\label{ddCLid}
		\operatorname{Div} (A_{\mbu} ) \de x=\mcD\mathrm{iv} (A_{\mathcal{H}}+A_{\boldsymbol \kappa} )\iota(\de x).
	\end{equation}
\end{Proposition}

\begin{proof}
The proof is essentially the same as for \pdese. In view of \eqref{Divid},
\[
0=\sum_\mbn\int \frac{\de}{\de t}( \Lrm \de x-L^{\boldsymbol \kappa}\iota(\de x))=\sum_\mbn\int \E_{u^\alpha}(\Lrm) (u^\alpha_{0;\mbzero})' \de x -\bigl\{\bigl(\mathcal{H}^\beta_\alpha\bigr)^{\boldsymbol{\dagger}}(\mathrm{E}_{\kappa^\beta}(L^{\boldsymbol \kappa}))\bigr\}\sigma^\alpha\iota(\de x).
\]	
This holds for arbitrary functions $(u^\alpha_{0;\mbzero})'$, which are independent at each base point $\mbn$. Therefore, \eqref{ddELid} follows, from which \eqref{ddinvEL} is obtained by invariantizing and using the independence of the functions $\sigma^\alpha$. Equation \eqref{ddCLid} is derived by equating the right-hand sides of \eqref{ddEL2} and \eqref{ddLk}, taking \eqref{ddELid} into account.
\end{proof}

In the original variables, there are two types of contribution to each Noether conservation law \eqref{dduNoe} for which $\xi_r\neq 0$. The first type arises from integration and summation by parts, so can be treated in much the same way as for \pdese. However, unlike the \pde case, $A_{\boldsymbol{\kappa}}$ cannot be neglected.
Each fundamental invariant $\kappa^\beta$ satisfies $\mbv_r\bigl(\kappa^\beta\bigr)=0$. However, only the vertical variables $[\mbu]$ depend on $t$, so the characteristic form of the generator is used. If $t=\varepsilon^r$, then~\smash{$\bigl(\kappa^\beta\bigr)'$}, evaluated at $\varepsilon^r=0$, reduces to
\begin{equation}\label{ddXkappa}
	X_r\bigl(\kappa^\beta\bigr)=\mbv_r\bigl(\kappa^\beta\bigr)-\xi_r D\bigl(\kappa^\beta\bigr)=-\xi_r D\bigl(\kappa^\beta\bigr)
	=-\mcD\bigl(\kappa^\beta\bigr) \iota(\xi_s) a_r^s(\rho).
\end{equation}
The last equality arises from \eqref{ddadxid}, evaluated on the moving frame $\rho$.

The counterpart of the remaining term in \eqref{dduNoe}, namely $D(\Lrm \xi_r)$, is derived as follows. From \eqref{ddadxi}~and~\eqref{Ldxinv},
\[
\Lrm \xi_r=\Lrm J_x^{-1}\widetilde{\xi}_s a_{r}^{s}(g)=\widetilde{\Lrm} \widetilde{\xi}_s a_{r}^{s}(g),
\]
for all $g\in G$; on the moving frame, this amounts to
\[
\Lrm \xi_r=L^{\boldsymbol{\kappa}}\iota(\xi_s) a_{r}^{s}(\rho).
\]
Consequently,
\begin{equation}\label{ddrem}
	D(\Lrm \xi_r) \de x=\mcD (L^{\boldsymbol{\kappa}}\iota(\xi_s) a_{r}^{s}(\rho) )\iota(\de x).
\end{equation}

Combining all of the above, we obtain the following formulation of the Noether conservation laws.

\begin{Proposition}\label{Prop: ddinvNoe}
Suppose that the conditions of Proposition {\rm \ref{ddEL}} hold. Then Noether's theorem gives the $R$ conservation laws
\begin{equation}
	\mcD\mathrm{iv}(A_r) \iota(\de x)=0,\qquad r=1,\dots,R,\label{ddCLaws}
\end{equation}
whose components are
\begin{align}
	&	A_r^0=H_{\alpha}^{0}\lbrace \iota (Q^\alpha_s )a^s_r(\rho)\rbrace-F_{\beta}^{0}\big\lbrace \mcD\bigl(\kappa^\beta\bigr) \iota(\xi_s) a^s_r(\rho)\big\rbrace+L^{\boldsymbol{\kappa}}\iota(\xi_s) a_{r}^{s}(\rho),\label{ddCL0}\\
	&	A_r^i=H_{\alpha}^{i}\lbrace \iota (Q^\alpha_s )a^s_r(\rho)\rbrace-F_{\beta}^{i}\big\lbrace \mcD\bigl(\kappa^\beta\bigr) \iota(\xi_s) a^s_r(\rho)\big\rbrace,\qquad i=1,\dots,m.\label{ddCLi}
\end{align}
\end{Proposition}

\begin{proof}
	For each $r$ in turn, replace $t$ by $\varepsilon^r$ and $u^\alpha$ by
	\begin{equation*}
		\widehat{u}^\alpha=u^\alpha+\varepsilon^r Q^\alpha_r(\mbn,x,[\mbu])+\mathcal{O}\bigl((\varepsilon^r)^2\bigr)
	\end{equation*}
	(prolonged as necessary), and evaluate the results at $\varepsilon^r=0$. Using the same reasoning as for~\pdese,
	\begin{equation*}
		\sigma^\alpha_r\big\vert_{\varepsilon^r=0}=\left\{\frac{\p (g\cdot u^\alpha)}{\p u^\beta} Q^\beta_r\right\}\bigg\vert_{g=\rho}
		=\bigl\{\widetilde{Q}^\alpha_s a^s_r(g)\bigr\}\bigg\vert_{g=\rho}=\iota(Q^\alpha_s)a^s_r(\rho),
	\end{equation*}
	which is substituted for $\sigma^\alpha$ in $A_\mcH$. The proof is completed by replacing $\bigl(\kappa^\beta\bigr)'$ in $A_{\boldsymbol{\kappa}}$ by the right-hand side of \eqref{ddXkappa}, and adding the remaining term \eqref{ddrem}.
\end{proof}

\begin{Corollary}
	Each component of the conservation laws \eqref{ddCLaws} is equivariant with respect to the moving frame $\rho$, because there exist functions $f^i_s$ of the invariants such that
\begin{equation*}
		A^i_r=f^i_s(\mbn,\iota(x),[\boldsymbol{\kappa}]) a_r^s(\rho),\qquad i=0,\dots, m,\quad r=1,\dots, R.
\end{equation*}
\end{Corollary}
\begin{proof}
	The invariant differential-difference operators in \eqref{ddCL0} and \eqref{ddCLi} act linearly on products of invariants and adjoint components $a_r^s(\rho)$. Consequently, every term in $A_r^i$ is of the form~$\phi^i_l a_r^l(\rho_{j;\mbK})$, where $\rho_{j;\mbK}=\mcD_{(j)}\s_\mbK\rho$ and $\phi^i_s$ is invariant. The invariantization of $\rho_{1;0}$ (in a~matrix representation) is the curvature matrix, $\rho_{1;0} \rho^{-1}$ (see \cite{gonccalves2016moving,mansfield2010practical}). By applying powers of~$\mcD$ to the curvature matrix and eliminating derivatives of order $1,\dots,j-1$, one finds that the term in braces in the identity below is invariant:
	\[
	\rho_{j;0}=\bigl\{\rho_{j;0} \rho^{-1}\bigr\}\rho.
	\]
	Shifting this, and using our earlier result that $\rho_{0;\mbK} \rho^{-1}$ is invariant for a difference moving frame, shows that
	\[
	\rho_{j;\mbK}=\bigl\{\rho_{j;\mbK} \rho_{0;\mbK}^{-1}\bigr\}\bigl\{\rho_{0;\mbK} \rho^{-1}\bigr\}\rho=\bigl\{\rho_{j;\mbK} \rho^{-1}\bigr\}\rho
	\]
	is a product of invariants (in braces) and $\rho$. The adjoint matrices $\bigl(a_r^l(g)\bigr)$ constitute a Lie group representation, so
	\[
	\phi^i_l a_r^l(\rho_{j;\mbK})=\bigl\{\phi^i_l a^l_s \bigl(\rho_{j;\mbK} \rho^{-1}\bigr)\bigr\}a_r^s(\rho),
	\]
	which is in the required form.
\end{proof}

\section{Examples}\label{exam}

\begin{Example}
To illustrate the calculations in a simple context (without any particular application), consider the Lagrangian one-form
	\begin{equation}\label{ddex1Lag}
	\Lrm \de x= \frac{(\uio)^2}{\uoi-\uoo} \de x,
	\end{equation}
whose Euler--Lagrange equation is
\begin{equation}\label{ddex1ELu}
\E_u(\Lrm)=\left(\frac{-2u_{2;0}}{\uoi-\uoo}+\frac{2\uio u_{1;1}-(\uio)^2}{(\uoi-\uoo)^2} - \frac{(u_{1;-1})^2}{(\uoo-u_{0;-1})^2}\right)=0.
\end{equation}
The one-form \eqref{ddex1Lag} is invariant under the two-parameter Lie group of point transformations
\[
g\colon\ (x,n,u)\longmapsto (\widetilde{x},n,\widetilde{u})=(bx,n,bu+a).
\]
The infinitesimal generators are linear combinations of $\mbv_1=\p_u$ and $\mbv_2=x\p_x+u\p_u$, which yield~$(\xi_1,Q_1)=(0,1)$ and $(\xi_2,Q_2)=(x,\uoo-x\uio)$. Then \eqref{dduNoe} gives the following conservation laws (expressed as one-forms for comparison):
\begin{gather}
-Q_1\E_u(\Lrm) \de x=D\left(\frac{2u_{1;0}}{\uoi-\uoo}\right)\de x+D_n\left(\frac{-(u_{1;-1})^2}{(\uoo-u_{0;-1})^2}\right)\de x=0,\nonumber\\
-Q_2\E_u(\Lrm) \de x=D\left(\frac{u_{1;0}(2\uoo-x\uio)}{\uoi-\uoo}\right)\de x+
	D_n\left(\frac{(u_{1;-1})^2(x\uio-\uoo)}{(\uoo-u_{0;-1})^2}\right)\de x=0.\label{dduCL2ex1}
\end{gather}
Note that \eqref{dduCL2ex1} includes the term $D(\Lrm \xi_2) \de x$.

Reflecting the identities $\mbv_1=b\widetilde{\mbv}$ and $\mbv_2=-a\widetilde{\mbv}_1+\widetilde{\mbv}_2$, the adjoint representation is given by
\begin{equation*}
	a_1^1(g)=b,\qquad a_1^2(g)=0,\qquad a_2^1(g)=-a,\qquad a_2^2(g)=1.
\end{equation*}

The normalization $\iota(x)=1,\ \iota(\uoo)=0$ gives a projectable moving frame $\rho$, on which
\[
a=\frac{-\uoo}{x} ,\qquad b=\frac{1}{x} .
\]
Therefore, $\iota(\de x)=x^{-1}\de x$ and the invariantized total derivative operator is $\mcD=xD$. On the moving frame, the adjoint representation has components
\begin{equation*}
	a_1^1(\rho)=\frac{1}{x} ,\qquad a_1^2(\rho)=0,\qquad a_2^1(\rho)=\frac{\uoo}{x} ,\qquad a_2^2(\rho)=1.
\end{equation*}
A generating set of invariants is
\[
\kappa^1=\iota(\uio)=\uio,\qquad\kappa^2=\iota(\uoi)=\frac{\uoi-\uoo}{x} ,
\]
which satisfy the syzygy
\begin{equation}\label{ddex1syz}
\kappa_{0;1}^1=\iota(u_{1;1})=\kappa^1+\kappa^2_{1;0}+\kappa^2.
\end{equation}
The Lagrangian one-form amounts to
\[
L^{\boldsymbol{\kappa}}\iota(\de x)=\frac{\bigl(\kappa^1\bigr)^2}{\kappa^2} \iota(\de x),
\]
and so
\[
\E_{\kappa^1}(L^{\boldsymbol{\kappa}})=\frac{2\kappa^1}{\kappa^2} ,\qquad \E_{\kappa^2}(L^{\boldsymbol{\kappa}})=-\left(\frac{\kappa^1}{\kappa^2}\right)^2.
\]
In terms of $\sigma=\iota(\uoo')=x^{-1}\uoo'$, the $t$-derivatives of the generating invariants are $\bigl(\kappa^\beta\bigr)'=\mcH^\beta\sigma$, where
\[
\mcH^1=\mcD+\id,\qquad \mcH^2=\s-\id.
\]
Consequently,
\begin{align*}
	\frac{\de L^{\boldsymbol{\kappa}}}{\de t} \iota(\de x)&=\left(\frac{2\kappa^1}{\kappa^2} \bigl(\kappa^1\bigr)'-\left(\frac{\kappa^1}{\kappa^2}\right)^2\bigl(\kappa^2\bigr)'\right)\iota(\de x)=\E_{\kappa^\beta}(L^{\boldsymbol{\kappa}})\bigl(\kappa^\beta\bigr)' \iota(\de x)\\
	&=\E_{\kappa^\beta}(L^{\boldsymbol{\kappa}}) \bigl(\mcH^\beta\sigma\bigr) \iota(\de x)=\bigl(\bigl(\mathcal{H}^{\beta}\bigr)^{\boldsymbol{\dagger}}\mathrm{E}_{\kappa^{\beta}} ( L^{\boldsymbol \kappa})\bigr)\sigma \iota(\de x)+\mcD\mathrm{iv}(A_\mcH) \iota(\de x),
\end{align*}
where
\begin{equation*}
	\mcD\mathrm{iv}(A_\mcH)=\mcD \bigl(\E_{\kappa^1}(L^{\boldsymbol{\kappa}}) \sigma\bigr)+D_n\bigl(\bigl\{\s^{-1}\E_{\kappa^2}(L^{\boldsymbol{\kappa}})\bigr\}\sigma\bigr).
\end{equation*}
In this example, $L^{\boldsymbol{\kappa}}$ does not involve derivatives or shifts of the generating invariants, and hence~$\mcD\mathrm{iv}(A_{\boldsymbol{\kappa}})=0$.
The invariantized Euler--Lagrange equation is
\begin{align*}
	0&=\bigl(\mathcal{H}^{\beta}\bigr)^{\boldsymbol{\dagger}} \mathrm{E}_{\kappa^{\beta}} ( L^{\boldsymbol \kappa})=(-\mcD+\id)\left(\frac{2\kappa^1}{\kappa^2}\right)+\bigl(\s^{-1}-\id\bigr)\left(-\left(\frac{\kappa^1}{\kappa^2}\right)^2\right)\\
	&=\frac{2\bigl(\kappa^1-\kappa^1_{1;0}\bigr)}{\kappa^2}+\frac{\kappa^1\bigl(\kappa^1+2\kappa^2_{1;0}\bigr)}{\bigl(\kappa^2\bigr)^2}-\left(\frac{\kappa^1_{0;-1}}{\kappa^2_{0;-1}}\right)^2.
\end{align*}
For comparison, one can invariantize \eqref{ddex1ELu} directly, using \eqref{ddex1syz} and
\[
\iota(u_{2;0})=\kappa^1_{1;0} ,\qquad \iota(u_{1;-1})=\kappa^1_{0;-1} ,\qquad \iota(u_{0;-1})=-\kappa^2_{0;-1} .
\]
From Proposition \ref{Prop: ddinvNoe}, the conservation laws given by Noether's theorem amount to
\begin{gather}
\mcD \bigl\{(\E_{\kappa^1}(L^{\boldsymbol{\kappa}}) \iota (Q_s )+L^{\boldsymbol{\kappa}}\iota(\xi_s)) a_{r}^{s}(\rho) \bigr\} \iota(\de x)\nonumber\\
\qquad{}+ D_n \bigl\{\bigl\{\s^{-1}\E_{\kappa^2}(L^{\boldsymbol{\kappa}})\bigr\} \iota (Q_s )a^s_r(\rho)\bigr\} \iota(\de x)=0.\label{ddex1CLform}
\end{gather}
Substituting
\[
\iota(Q_s) a_1^s(\rho)=\frac{1}{x} ,\qquad \iota(Q_s) a_2^s(\rho)=\frac{\uoo}{x} -\kappa^1,\qquad
\iota(\xi_s) a_1^s(\rho)=0,\qquad
\iota(\xi_s) a_2^s(\rho)=1,
\]
into \eqref{ddex1CLform} gives the conservation laws
\begin{align*}
	0&=\mcD \left\{\frac{2\kappa^1}{x\kappa^2}\right\} \iota(\de x)+ D_n \left\{- \frac{1}{x}\left(\frac{\kappa^1_{0;-1}}{\kappa^2_{0;-1}}\right)^2\right\} \iota(\de x),\\
	0&=\mcD \left\{\frac{2\kappa^1}{\kappa^2}\left(\frac{\uoo}{x}-\kappa^1\right)+\frac{\bigl(\kappa^1\bigr)^2}{\kappa^2}\right\} \iota(\de x)+ D_n \left\{\left(\frac{\kappa^1_{0;-1}}{\kappa^2_{0;-1}}\right)^2\left(\kappa^1-\frac{\uoo}{x}\right)\right\} \iota(\de x)\\\
	&=\mcD \left\{\frac{\kappa^1}{\kappa^2}\left(\frac{2\uoo}{x}-\kappa^1\right)\right\} \iota(\de x)+ D_n \left\{\left(\frac{\kappa^1_{0;-1}}{\kappa^2_{0;-1}}\right)^2\left(\kappa^1-\frac{\uoo}{x}\right)\right\} \iota(\de x).
\end{align*}
\end{Example}

\begin{Example}
	Method of lines semi-discretizations are a common source of \ddes with just one continuous independent variable. The nonlinear Schr\"odinger (NLS) equation for a field with real and imaginary parts $u$ and $v$ respectively has the following (non-integrable) semi-discretization, with uniform step length $h$:
	\begin{gather*}
		-\vio+\uoo\bigl(\uoo^2+\voo^2\bigr)+h^{-2}(\uom-2\uoo+\uoi)=0,\\
		\uio+\voo\bigl(\uoo^2+\voo^2\bigr)+h^{-2}(\vom-2\voo+\voi)=0.
	\end{gather*}
	These are the Euler--Lagrange equations corresponding to the Lagrangian one-form
	\begin{gather*}
	\Lrm \de x=\left\{\frac{1}{2}(\voo\uio-\uoo\vio)+\frac{1}{4}\bigl(\uoo^2+\voo^2\bigr)^2\right.\\
\left.\phantom{\Lrm \de x=}{}-\frac{1}{2}h^{-2}\bigl((\uoi-\uoo)^2+(\voi-\voo)^2\bigr)\right\}\de x,
	\end{gather*}
	which is invariant under the two-parameter abelian Lie group of point transformations
	\[
	g\colon\ (x,n,u,v)\longmapsto(\widetilde{x},\widetilde{n},\widetilde{u},\widetilde{v})=(x+a,n,u\cos b+v\sin b,-u\sin b+v\cos b).
	\]
	The infinitesimal generators are $\mbv_1=\p_x$ and $\mbv_2=v\p_u-u\p_v$, so (using variable names rather than indices, for clarity)
	\begin{equation*}
		(\xi_1,Q_1^u,Q_1^v)=(1,-\uio,-\vio),\qquad (\xi_2,Q_2^u,Q_2^v)=(0,\voo,-\uoo).
	\end{equation*}
	As the Lie group is abelian, the adjoint representation is the identity, so $a_r^s(g)=\delta_r^s$ for all $g$.
	
	We now choose the normalization $\iota(x)=0$, $\iota(\voo)=0$, temporarily restricting attention to~$\uoo>0$. (Other normalizations can be used for the remaining coordinate patches.) This gives the frame $\rho$ defined by
	\[
	a\big|_\rho = -x,\qquad b\big|_\rho=\tan^{-1}\left(\frac{\voo}{\uoo}\right).
	\]
	In the calculations that follow, $\cos b$ and $\sin b$ (but not $b$) are evaluated on the frame, so we use
	\[
	a\big|_\rho = -x,\qquad \cos b\big|_\rho=\frac{\uoo}{\sqrt{\uoo^2+\voo^2}} ,\qquad \sin b\big|_\rho=\frac{\voo}{\sqrt{\uoo^2+\voo^2}} ,
	\]
	which extends to other coordinate patches. Note that $\iota(\de x)=\de x$, and so the invariantized total derivative is $\mcD=D$.
	The invariants are generated by
	\begin{gather*}
		\kappa^1=\iota(\uoo)=\sqrt{\uoo^2+\voo^2},\qquad \kappa^2=\uoo\vio-\voo\uio,\qquad
\kappa^3=\uoo\uoi+\voo\voi.
	\end{gather*}
	To see this, note that all derivatives can be obtained from
	\[
	\uio=\frac{\uoo\kappa^1\kappa^1_{1;0}-\voo\kappa^2}{\bigl(\kappa^1\bigr)^2} ,\qquad \vio=\frac{\uoo\kappa^2+\voo\kappa^1\kappa^1_{1;0}}{\bigl(\kappa^1\bigr)^2} ,
	\]
	and that under the constraint $\iota(\voi)\geq 0$, all shifts can be obtained from
	\begin{equation*}
	\uoi=\frac{\uoo\kappa^3-\voo\phi}{\bigl(\kappa^1\bigr)^2} ,\qquad \voi=\frac{\uoo\phi+\voo\kappa^3}{\bigl(\kappa^1\bigr)^2} ,
\end{equation*}
where \[\phi=\bigl\{\bigl(\kappa^1\kappa^1_{0;1}\bigr)^2-\bigl(\kappa^3\bigr)^2\bigr\}^{1/2}.\]
	We adopt this constraint for definiteness; if it is not satisfied, replace $\phi$ by $-\phi$ throughout. By calculating $u_{1;1}$ (or $v_{1;1}$), one obtains the syzygy
	\begin{equation*}
		\frac{\kappa^3\kappa^1_{1;1}}{\kappa^1_{0;1}}-\kappa^3_{1;0}+\frac{\kappa^1_{1;0}}{\kappa^1}+\frac{\phi \kappa^2}{\bigl(\kappa^1\bigr)^2}-\frac{\phi \kappa^2_{0;1}}{\bigl(\kappa^1_{0;1}\bigr)^2}=0.
	\end{equation*}

In terms of the generating parametric derivatives,
\[
\sigma^u=\iota(\uoo')=\frac{\uoo\uoo'+\voo\voo'}{\kappa^1} ,\qquad \sigma^v=\iota(\voo')=\frac{\uoo\voo'-\voo\uoo'}{\kappa^1} ,
\]
the derivatives of the generating invariants are
\begin{gather*}
\bigl(\kappa^1\bigr)'=\sigma^u,\qquad
 \bigl(\kappa^2\bigr)'=\frac{2\kappa^2}{\kappa^1} \sigma^u+\bigl(\kappa^1 D-\kappa^1_{1;0}\bigr) \sigma^v,\\ \bigl(\kappa^3\bigr)'=\left(\frac{\kappa^3}{\kappa^1}+\frac{\kappa^3}{\kappa^1_{0;1}} \s\right)\sigma^u+\left(\frac{\phi}{\kappa^1}-\frac{\phi}{\kappa^1_{0;1}} \s\right) \sigma^v.
\end{gather*}
As $\iota(\de x)=\de x$ and $\mcD=D$, the invariant Euler--Lagrange equations and conservation laws can be calculated directly from
\[
L^{\boldsymbol{\kappa}}=- \frac{1}{2}\kappa^2+\frac{1}{4}\bigl(\kappa^1\bigr)^4-\frac{1}{2}h^{-2}\bigl(\bigl(\kappa^1_{0;1}\bigr)^2-2\kappa^3+\bigl(\kappa^1\bigr)^2\bigr).
\]
Differentiating this, we obtain
\begin{align}
	\frac{\de}{\de t}L^{\boldsymbol{\kappa}}={}&- \frac{1}{2}\bigl(\kappa^2\bigr)'+\bigl(\kappa^1\bigr)^3\bigl(\kappa^1\bigr)'-h^{-2}\bigl(\kappa^1_{0;1}\bigl(\kappa^1_{0;1}\bigr)'-\bigl(\kappa^3\bigr)'+\kappa^1\bigl(\kappa^1\bigr)'\bigr)\nonumber\\
	={}&\bigl\{\bigl(\kappa^1\bigr)^3-2h^{-2}\kappa^1\bigr\}\bigl(\kappa^1\bigr)'- \frac{1}{2}\bigl(\kappa^2\bigr)'+h^{-2}\bigl(\kappa^3\bigr)'\label{ddex2ibp}
	+\underbrace{D_n \bigl(-h^{-2}\kappa^1\bigl(\kappa^1\bigr)'\bigr)}_{\mcD\mathrm{iv}A_{\boldsymbol{\kappa}}}\\
	={}&\left\{\bigl(\kappa^1\bigr)^3-\frac{2\kappa^1}{h^2}-\frac{\kappa^2}{\kappa^1}+\frac{\kappa^3}{h^2\kappa^1}+\frac{\kappa^3}{h^2\kappa^1_{0;1}} \s\right\}\sigma^u\nonumber\\
	&+\left\{-\frac{1}{2}\kappa^1 D+\frac{1}{2}\kappa^1_{1;0}+\frac{\phi}{h^2\kappa^1}-\frac{\phi}{h^2\kappa^1_{0;1}} \s\right\} \sigma^v+D_n \bigl(-h^{-2}\kappa^1\bigl(\kappa^1\bigr)'\bigr)\nonumber\\ ={}&\left\{\bigl(\kappa^1\bigr)^3-\frac{2\kappa^1}{h^2}-\frac{\kappa^2}{\kappa^1}+\frac{\kappa^3}{h^2\kappa^1}+\frac{\kappa^3_{0,-1}}{h^2\kappa^1}\right\}\sigma^u+\left\{\kappa^1_{1;0}+\frac{\phi}{h^2\kappa^1}-\frac{\s^{-1}\phi}{h^2\kappa^1}\right\} \sigma^v\nonumber\\
	&+D\left(-\frac{1}{2}\kappa^1\sigma^v\right)+D_n \left(\frac{\kappa^3_{0,-1}}{h^2\kappa^1} \sigma^u-\frac{\s^{-1}\phi}{h^2\kappa^1} \sigma^v\right)+D_n \bigl(-h^{-2}\kappa^1\bigl(\kappa^1\bigr)'\bigr).\nonumber
	\end{align}
Consequently, the invariantized Euler--Lagrange equations are
\begin{align*}
	\iota(\E_u(\Lrm))&=\bigl(\kappa^1\bigr)^3-\frac{2\kappa^1}{h^2}-\frac{\kappa^2}{\kappa^1}+\frac{\kappa^3}{h^2\kappa^1}+\frac{\kappa^3_{0,-1}}{h^2\kappa^1}=0,\nonumber\\
	\iota(\E_v(\Lrm))&=\kappa^1_{1;0}+\frac{\phi}{h^2\kappa^1}-\frac{\s^{-1}\phi}{h^2\kappa^1}=0.
\end{align*}
In this example, the conservation laws have contributions from both $A_{\boldsymbol{\kappa}}$ (see \eqref{ddex2ibp}) and $A_\mcH$. Using $a_r^s(\rho)=\delta_r^s$ and
\begin{gather*}
\begin{split}
&\iota(Q_1^u)=-\kappa^1_{1;0},\qquad \iota(Q_1^v)=-\kappa^2/\kappa^1,\qquad \iota(Q_2^u)=0,\\
& \iota(Q_2^v)=-\kappa^1,\qquad\iota(\xi_1)=1,\qquad\iota(\xi_2)=0,
\end{split}
\end{gather*}
we obtain the following conservation laws from Proposition \ref{Prop: ddinvNoe}:
\begin{align*}
	0={}&D\left\{-\frac{1}{2}\kappa^1\iota(Q_1^v)+\iota(\xi_1)L^{\boldsymbol{\kappa}}\right\}+D_n \left\{\frac{\kappa^3_{0,-1}}{h^2\kappa^1} \iota(Q_1^u)-\frac{\s^{-1}\phi}{h^2\kappa^1} \iota(Q_1^v)-\frac{\kappa^1}{h^2}\bigl(-\iota(\xi_1)\kappa^1_{1;0}\bigr)\right\}\\
	={}&D\left\{\frac{1}{4}\bigl(\kappa^1\bigr)^4-\frac{\bigl(\kappa^1_{0;1}\bigr)^2}{2h^2}+\frac{\kappa^3}{h^2}-\frac{\bigl(\kappa^1\bigr)^2}{2h^2}\right\}+D_n \left\{-\frac{\kappa^1_{1;0}\kappa^3_{0,-1}}{h^2\kappa^1}+\frac{\kappa^2 \s^{-1}\phi}{h^2\bigl(\kappa^1\bigr)^2}+\frac{\kappa^1\kappa^1_{1;0}}{h^2}\right\},\\
	0={}&D\left\{-\frac{1}{2}\kappa^1\iota(Q_2^v)\right\}+D_n \left\{\frac{\kappa^3_{0,-1}}{h^2\kappa^1} \iota(Q_2^u)-\frac{\s^{-1}\phi}{h^2\kappa^1} \iota(Q_2^v)\right\}\\
	={}&D\left\{\frac{1}{2}\bigl(\kappa^1\bigr)^2\right\}+D_n \left\{\frac{\s^{-1}\phi}{h^2}\right\}.
\end{align*}
\end{Example}

\section{Concluding remarks}\label{conc}

For \pdes and \ddese, the prolongation space over a fixed base point $\mbn$ provides a continuous setting in which moving frames can be used. Difference moving frames respect the ordering of each discrete independent variable and the arbitrariness of the base point. For variational problems, we have shown how to calculate the invariantized Euler--Lagrange equations and equivariant Noether conservation laws directly from an invariant Lagrangian, $L^{\boldsymbol{\kappa}}$.

We have treated the coordinates $n^i$ on the lattice of independent variables, $\bbz^m$, as being fixed by the Lie group of transformations. For a formulation that allows discrete symmetries of the lattice, it would be necessary to replace the Lagrangian $L^{\boldsymbol{\kappa}}$ by the Lagrangian $m$-form, $L^{\boldsymbol{\kappa}} \mathrm{vol}$, where $\mathrm{vol}$ denotes the difference volume form (see \cite{hydon2004variational,mansfield2008difference}). This adds complexity, but little extra insight. However, when each $n^i$ is fixed, only the coefficients of difference forms are transformed by the group action, so one can use moving frames as we have done, without the additional machinery of difference forms.

Our treatment of \ddes has been restricted to a single continuous independent variable and a projectable moving frame, enabling the shift and invariant differential operators to commute. More generally, let $p$ be the number of continuous independent variables. If $p>1$, the complexity increases, because the invariant differential operators do not necessarily commute with one another. However, if the moving frame is projectable, all shift operators commute with the invariant derivatives, enabling existing results from PDE theory to be used for the differential part of the calculations. The requirement for the group action to be projectable is not sufficient to guarantee the existence of a projectable moving frame. Nevertheless, projectable moving frames are relevant to many \ddes of interest, including some well-known integrable systems and method of lines semi-discretizations of PDEs.

\subsection*{Acknowledgments}

We thank Professor Elizabeth Mansfield, whose strong advocacy of moving frames, insight and encouragement have led to this project. The work was partially supported by EPSRC grant number EP/R513246/1.

\pdfbookmark[1]{References}{ref}
\LastPageEnding


\begin{thebibliography}{99}
\footnotesize\itemsep=0pt

\bibitem{boutin2002orbit}
Boutin M., On orbit dimensions under a simultaneous {L}ie group action on {$n$}
 copies of a manifold, \textit{J.~Lie Theory} \textbf{12} (2002), 191--203,
 \href{https://arxiv.org/abs/math-ph/0009021}{arXiv:math-ph/0009021}.

\bibitem{fels1998moving}
Fels M., Olver P.J., Moving coframes.~{I}. {A}~practical algorithm,
 \href{https://doi.org/10.1023/A:1005878210297}{\textit{Acta Appl. Math.}} \textbf{51} (1998), 161--213.

\bibitem{fels1999moving}
Fels M., Olver P.J., Moving coframes.~{II}. {R}egularization and theoretical
 foundations, \href{https://doi.org/10.1023/A:1006195823000}{\textit{Acta Appl. Math.}} \textbf{55} (1999), 127--208.

\bibitem{gonccalves2012moving}
Gon\c{c}alves T.M.N., Mansfield E.L., On moving frames and {N}oether's
 conservation laws, \href{https://doi.org/10.1111/j.1467-9590.2011.00522.x}{\textit{Stud. Appl. Math.}} \textbf{128} (2012), 1--29,
 \href{https://arxiv.org/abs/1006.4660}{arXiv:1006.4660}.

\bibitem{gonccalves2013moving}
Gon\c{c}alves T.M.N., Mansfield E.L., Moving frames and conservation laws for
 {E}uclidean invariant {L}agrangians, \href{https://doi.org/10.1111/j.1467-9590.2012.00566.x}{\textit{Stud. Appl. Math.}} \textbf{130}
 (2013), 134--166, \href{https://arxiv.org/abs/1106.3964}{arXiv:1106.3964}.

\bibitem{gonccalves2016moving}
Gon\c{c}alves T.M.N., Mansfield E.L., Moving frames and {N}oether's
 conservation laws~-- the general case, \href{https://doi.org/10.1017/fms.2016.24}{\textit{Forum Math. Sigma}} \textbf{4}
 (2016), e29, 55~pages, \href{https://arxiv.org/abs/1306.0847}{arXiv:1306.0847}.

\bibitem{hydon2014difference}
Hydon P.E., Difference equations by differential equation methods,
 \textit{Cambridge Monogr. Appl. Comput. Math.}, Vol.~27, \href{https://doi.org/10.1017/CBO9781139016988}{Cambridge University
 Press}, Cambridge, 2014.

\bibitem{hydon2004variational}
Hydon P.E., Mansfield E.L., A variational complex for difference equations,
 \href{https://doi.org/10.1007/s10208-002-0071-9}{\textit{Found. Comput. Math.}} \textbf{4} (2004), 187--217.

\bibitem{kimolver}
Kim P., Olver P.J., Geometric integration via multi-space, \href{https://doi.org/10.1070/RD2004v009n03ABEH000277}{\textit{Regul.
 Chaotic Dyn.}} \textbf{9} (2004), 213--226.

\bibitem{kogan2003invariant}
Kogan I.A., Olver P.J., Invariant {E}uler--{L}agrange equations and the
 invariant variational bicomplex, \href{https://doi.org/10.1023/A:1022993616247}{\textit{Acta Appl. Math.}} \textbf{76}
 (2003), 137--193.

\bibitem{kupershmidt1985discrete}
Kuperschmidt B.A., Discrete {L}ax equations and differential-difference
 calculus, \textit{Ast\'erisque} \textbf{123} (1985), 212~pages.

\bibitem{mansfield2010practical}
Mansfield E.L., A practical guide to the invariant calculus, \textit{Cambridge
 Monogr. Appl. Comput. Math.}, Vol.~26, \href{https://doi.org/10.1017/CBO9780511844621}{Cambridge University Press}, Cambridge,
 2010.

\bibitem{mansfield2008difference}
Mansfield E.L., Hydon P.E., Difference forms, \href{https://doi.org/10.1007/s10208-007-9015-8}{\textit{Found. Comput. Math.}}
 \textbf{8} (2008), 427--467.

\bibitem{mansfield2013discrete}
Mansfield E.L., Mar\'{\i}~Beffa G., Wang J.P., Discrete moving frames and
 discrete integrable systems, \href{https://doi.org/10.1007/s10208-013-9153-0}{\textit{Found. Comput. Math.}} \textbf{13}
 (2013), 545--582, \href{https://arxiv.org/abs/1212.5299}{arXiv:1212.5299}.

\bibitem{mansfield2019movinga}
Mansfield E.L., Rojo-Echebur\'ua A., Hydon P.E., Peng L., Moving frames and
 {N}oether's finite difference conservation laws~{I}, \href{https://doi.org/10.1093/imatrm/tnz004}{\textit{Trans. Math.
 Appl.}} \textbf{3} (2019), tnz004, 47~pages, \href{https://arxiv.org/abs/1804.00317}{arXiv:1804.00317}.

\bibitem{mansfield2019movingb}
Mansfield E.L., Rojo-Echebur\'ua A., Moving frames and {N}oether's finite
 difference conservation laws~{II}, \href{https://doi.org/10.1093/imatrm/tnz005}{\textit{Trans. Math. Appl.}} \textbf{3}
 (2019), tnz005, 26~pages, \href{https://arxiv.org/abs/1808.03606}{arXiv:1808.03606}.

\bibitem{maribeffa2018discrete}
Mar\'{\i}~Beffa G., Mansfield E.L., Discrete moving frames on lattice varieties
 and lattice-based multispaces, \href{https://doi.org/10.1007/s10208-016-9337-5}{\textit{Found. Comput. Math.}} \textbf{18}
 (2018), 181--247.

\bibitem{olver2000applications}
Olver P.J., Applications of {L}ie groups to differential equations, 2nd ed.,
 \textit{Grad. Texts in Math.}, Vol.~107, \href{https://doi.org/10.1007/978-1-4612-4350-2}{Springer}, New York, 1993.

\bibitem{olvermulti}
Olver P.J., Geometric foundations of numerical algorithms and symmetry, \href{https://doi.org/10.1007/s002000000053}{\textit{Appl. Algebra Engrg. Comm. Comput.}}, Vol.~11, 2001,
 417--436.

\bibitem{olvercompvis}
Olver P.J., Sapiro G., Tannenbaum A., Differential invariant signatures and
 flows in computer vision: {A}~symmetry group approach, in Geometry-{D}riven
 {D}iffusion in {C}omputer {V}ision, \textit{Comput. Imaging Vis.}, Vol.~1,
 Editor B.M. ter Haar~Romeny, \href{https://doi.org/10.1007/978-94-017-1699-4_11}{Kluwer Academic Publishers}, Dordrecht, 1994,
 255--306.

\bibitem{peng2022transformations}
Peng L., Hydon P.E., Transformations, symmetries and {N}oether theorems for
 differential-difference equations, \href{https://doi.org/10.1098/rspa.2021.0944}{\textit{Proc.~A.}} \textbf{478} (2022),
 20210944, 17~pages, \href{https://arxiv.org/abs/2112.06030}{arXiv:2112.06030}.

\end{thebibliography}
\end{document}